\documentclass[prx,aps,floatfix,amsmath,superscriptaddress,twocolumn]{revtex4}
\usepackage{amssymb,enumerate}
\usepackage{graphicx}
\usepackage{graphics}
\usepackage{amsmath}
\usepackage{amsthm,bbm}
\usepackage{color}
\usepackage{dsfont}
\usepackage{hyperref}

\usepackage{lipsum}
\usepackage{lmodern}
\usepackage{tcolorbox}

\usepackage{amsfonts}
\usepackage{graphicx,graphics,epsfig,times,bm,bbm,amssymb,amsmath,amsfonts,mathrsfs}
\usepackage[normalem]{ulem}
\usepackage{setspace}
\usepackage{subfigure}
\usepackage{dsfont}
\usepackage{braket}
\usepackage{upgreek }
\usepackage{tikz}
\usepackage{natbib}
\usepackage{chngcntr}

\newtheorem{theorem}{Theorem}

\newtheorem{lemma}[theorem]{Lemma}

\newtheorem{proposition}[theorem]{Proposition}
\newtheorem{remark}[theorem]{Remark}

\newcommand{\bes} {\begin{subequations}}
\newcommand{\ees} {\end{subequations}}
\newcommand{\bea} {\begin{eqnarray}}
\newcommand{\eea} {\end{eqnarray}}
\newcommand{\be} {\begin{equation}}
\newcommand{\ee} {\end{equation}}

\def\>{\rangle}
\def\<{\langle}
\def\Tr{\textrm{Tr}}

\newcommand{\ignore}[1]{}

\begin{document}	
	\title{Operational Interpretation of Quantum Fisher Information in Quantum Thermodynamics }
	 
	
	\author{Iman Marvian}
\affiliation{Departments of Physics \& Electrical and Computer Engineering, Duke University, Durham, North Carolina 27708, USA}

	\begin{abstract}
In the framework of quantum thermodynamics preparing a  quantum system in a general state requires the consumption of two distinct resources, namely, work and coherence. It has been shown that the work cost of preparing a quantum state is determined by its free energy. Considering a similar setting, here we determine the coherence cost of preparing a general state when there are no restrictions on work consumption. More precisely, the coherence cost is defined as the minimum rate of consumption of systems in a pure coherent state, that is needed to prepare copies of the desired system.  We show that the coherence cost of any system is determined by its quantum Fisher information about the time parameter, hence introducing a new operational interpretation of this central quantity of quantum metrology.  Our resource-theoretic approach also reveals a previously unnoticed connection between two fundamental properties of quantum Fisher information.
	\end{abstract}
		
		\maketitle

Information-theoretic approach  to quantum thermodynamics  and,  more specifically, the resource-theoretic approach \cite{chitambar2019quantum}  has proven to be extremely fruitful.  This, for instance, 
 has lead to the discovery of new aspects of quantum coherence in thermodynamics (See, e.g., \cite{lostaglio2015description,lostaglio2015quantumPRX,korzekwa2016extraction, narasimhachar2015low, marvian2020coherence,  marvian2019no, lostaglio2019coherence, streltsov2017colloquium}).   In this approach, which is partly inspired by the entanglement theory, one studies the inter-convertability of systems under a limited set of operations, which presumably can be implemented with negligible thermodynamic costs (this assumption relies on certain idealizations about available resources and achievable control of quantum systems).    A popular 
choice  is the set of thermal operations, i.e.,  those  that can be implemented by coupling the  system to a thermal bath via energy-conserving unitaries \cite{janzing2000thermodynamic, brandao2013resource}.   

From a thermodynamics point of view,  preparing  a general quantum state requires consumption of both work and energetic coherence, i.e., coherence between states with different energies, which can also be understood as asymmetry with respect to time translations \cite{marvian2014modes, lostaglio2015quantumPRX, marvian2016quantum,marvian2016quantify}. In the resource-theoretic framework of quantum thermodynamics, it has been shown that the work cost of preparing   many independent and identically distributed (iid) copies of any  quantum system 
is determined by its free energy   \cite{brandao2013resource}.  
On the other hand,  characterizing the coherence cost of preparing  quantum systems has remained an open question \cite{winter2016operational, streltsov2017colloquium}. 


In this Letter we settle this question and show that the coherence cost of preparing  a quantum system in a general state is determined by the   Quantum Fisher Information (QFI) \cite{Helstrom:book, Holevo:book, BraunsteinCaves:94} of the system about the time parameter.  More precisely, to prepare copies of the desired system in the iid regime, the minimum rate of consumption of systems in a fixed pure coherent state is determined by the ratio of QFI's of the desired system to the input pure system (See Fig.\ref{Fig1}).  Interestingly, a similar result does not hold for the reverse process, called  coherence distillation:  for generic mixed input states the rate of conversion to pure coherent states is zero  \cite{marvian2020coherence}.  

Hence, our result reveals a novel operational interpretation of QFI, which is the  central quantity of  quantum metrology \cite{Giovannetti:05, giovannetti2011advances}.  
Remarkably, our resource-theoretic approach also   clarifies a close connection between two different fundamental properties of QFI, namely QFI as a convex roof of variance and QFI as the variance of purification of state.  While QFI has been extensively studied in quantum metrology, to our knowledge this connection has not been appreciated before.

 To focus on coherence as a resource independent of work,  one can supplement thermal operations  with a battery or work reservoir that can provide an  unlimited amount of work (In other words, one can make work a free resource). It turns out \cite{keyl1999optimal, proof_thesis, marvian2020coherence} that in this way one can  implement all and only  time-Translationally Invariant (TI) operations, i.e.,  completely positive trace-preserving maps satisfying  the covariance condition,
\be\label{cov}
e^{-i H_\text{out} t}\  \mathcal{E}_{\text{TI}}(\sigma)\ e^{i H_\text{out} t}=\mathcal{E}_{\text{TI}}\big(e^{-i H_\text{in} t}\sigma e^{i H_\text{in} t}\big)\ ,
\ee
for all density operators $\sigma$ and all times $t$ \cite{QRF_BRS_07, gour2008resource, marvian2013theory, Marvian_thesis}. Here,   $ H_\text{in}$ and $H_\text{out}$ are, respectively, the input and output  Hamiltonians. TI operations can not generate (energetic) coherence:  to prepare systems containing coherence via TI operations, one needs an input that contains  coherence. On the other hand, preparing incoherent states, i.e., those that commute with the system Hamiltonian, does not require consuming coherence.   In summary, to understand  coherence as a resource independent of work, we  study state conversions under TI operations. It is also worth noting that going beyond TI operations makes  coherence a free resource: 
using any non-TI operation it is possible to generate energetic coherence 
from incoherent states, albeit this may require correlation between the input of the operation and an auxiliary system \cite{marvian2020coherence}.

TI operations and the notion of coherence cost also arise in the study of quantum clocks.  While  coherent states and non-TI operations  should be defined relative to a background reference clock,  Eq.(\ref{cov}) means that TI operations can be defined and implemented without access to this  clock 
  \cite{QRF_BRS_07, gour2008resource, marvian2020coherence}.  Suppose  one does not have access to the  reference clock, but is given  quantum clocks that are synchronized with it.   What is the minimum rate of consumption of quantum clocks in pure states, that is needed to prepare copies of a desired system  (See Fig.\ref{Fig1})?  
   Again, we find that the answer is given by the QFI of the system about the time parameter.   \\

\begin{figure} [h]
\begin{center}
\includegraphics[scale=.65]{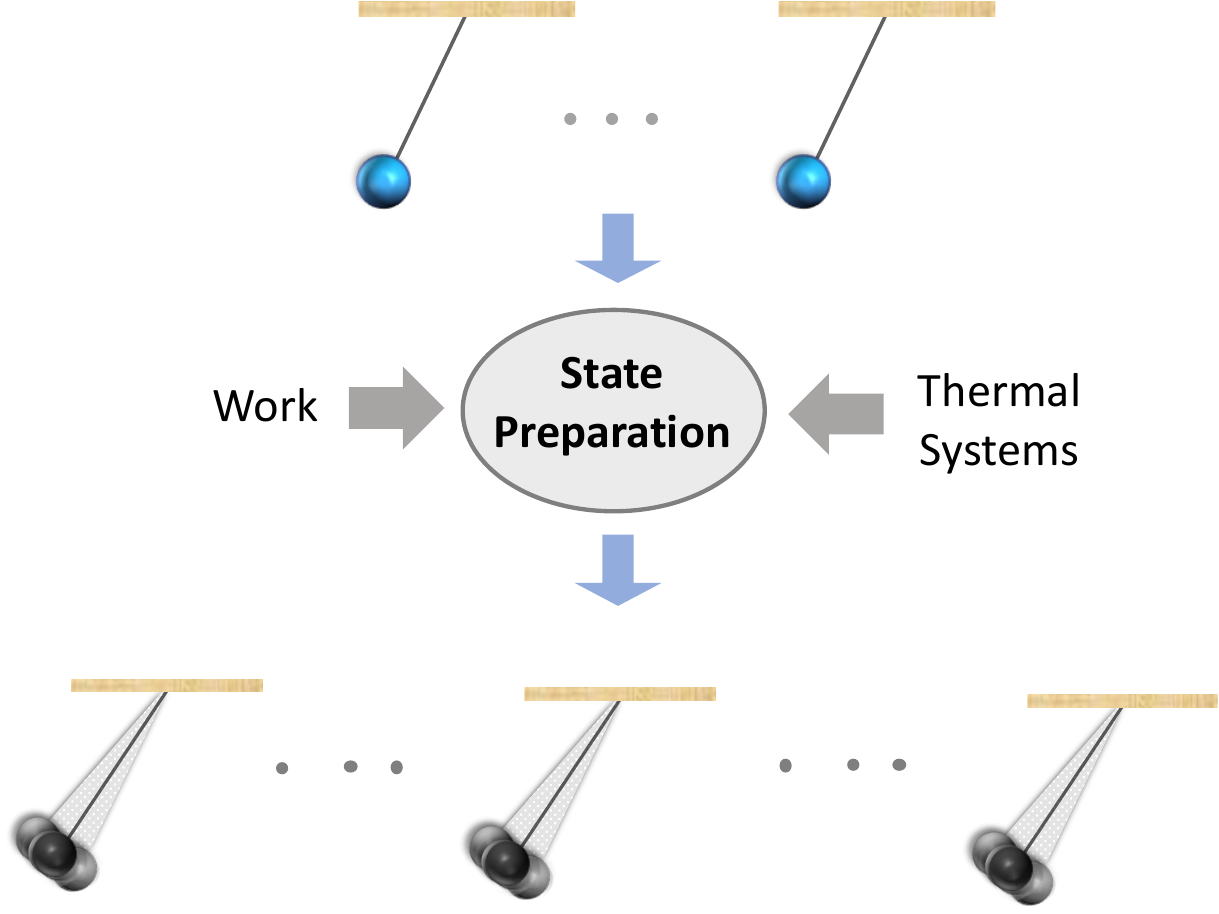}
\caption{Preparing a  quantum system in a general state requires consumption of both work and coherence. Here, we study the coherence cost of preparing state, when there are no limitations on work consumption.  Equivalently, we characterize the minimum rate of consumption of quantum clocks that is needed to prepare a general state, when one does not have access to the standard reference clock. }\label{Fig1}
\end{center}
\end{figure}
\noindent \emph{Pure states in the iid regime}---We study systems with finite-dimensional Hilbert spaces. Each system is specified by its Hamiltonian $H$ and density operator $\rho$.  We assume the systems under consideration have periodic dynamics with a fixed but arbitrary period $\tau$ such that $\tau=\inf\{t>0: e^{-i H t}\rho e^{i H t}=\rho\}$.   Under TI operations, a system with period $\tau$ can  only be converted to systems with period $\tau/k$, for an integer $k$. In the following, we consider $n$ copies of a system with Hamiltonian $H$ and state $\rho$,  which means their joint state is $\rho^{\otimes n}$ and their total Hamiltonian is $\sum_{j=0}^{n-1} I^{\otimes j}\otimes H \otimes I^{\otimes (n-j-1)}$.

Consider many copies of a system with Hamiltonian $H_1$, pure state $\psi_1$ and period $\tau$.  Is it possible to convert these systems to many copies of another system with the same period $\tau$, in pure state $\psi_2$ and Hamiltonian $H_2$,  using  TI operations? Since exact conversions are often 
impossible and physically intractable, as usual we allow a vanishing  error quantified, e.g., in terms of the trace distance $D(\rho,\sigma)=\|\rho-\sigma\|_1/2$.  
In the following,      $V_H(\psi)=\langle\psi|H^2|\psi\rangle-\langle\psi|H|\psi\rangle^2$  denotes the energy variance of pure state $\psi$ with respect to Hamiltonian $H$.  Our first main result is    
\begin{theorem}\label{Thmold} 
Consider a pair of systems with pure states $\psi_1$ and $\psi_2$ and Hamiltonians $H_1$ and $H_2$, with equal periods.  Using TI operations the state conversion 
\be\nonumber
 |\psi_1\rangle^{\otimes n}\xrightarrow{TI}\stackrel{\epsilon_n}{ \approx}  |\psi_2\rangle^{\otimes \lceil R n\rceil} \ \ \ \ \ \ \   \text{as } n\rightarrow\infty\ ,\ \   \epsilon_n\rightarrow 0\ ,
\ee
with vanishing error $\epsilon_n$ in trace distance is possible if rate $R\le V_{H_1}(\psi_1)/V_{H_2}(\psi_2)$ and is impossible  if $R> V_{H_1}(\psi_1)/V_{H_2}(\psi_2)$. 
 \end{theorem}
Hence, in the iid regime   oscillators in pure states with the same frequencies are equivalent resources, in the sense that by adding or absorbing sufficient amount of energy their coherence content, or equivalently, their information content about time,  can be converted from one form to another. Note that  the maximal achievable rate from system 1 to 2, namely  $V_{H_1}(\psi_1)/V_{H_2}(\psi_2)$,  is the inverse of the maximal rate from system 2 to 1. In this sense the process is reversible.    Consequently,  in this regime the usefulness of a clock can be quantified by a single number, namely its energy variance.  In other words, we can pick a standard \emph{clock-bit} (coherence-bit) or \emph{c-bit} with period $\tau$ and quantify the amount of resource of a general state relative to this standard. A convenient choice is a two-level system with Hamiltonian $H_\text{c-bit}=\pi \sigma_z/\tau$ and state $|\Theta\rangle_{\text{c-bit}}=(|0\rangle+|1\rangle)/\sqrt{2}$, with the energy variance $\pi^2/\tau^2$.

This theorem, which is proven in the Supplementary Material (SM),   strengthens and generalizes a previously known result  \cite{schuch2004nonlocal, schuch2004quantum, gour2008resource}  in multiple ways. The common intuition behind all these results, first discussed in \cite{schuch2004nonlocal, schuch2004quantum}, is based on the Central Limit Theorem which implies that the total energy distribution of many copies of a state converges to a Gaussian distribution, and hence is characterized by its variance and mean, which are both additive.  Then, as the mean energy can be changed arbitrarily by TI operations, the conversion rate is determined by the ratio of  variances.

 One aspect of theorem  \ref{Thmold} that makes it stronger than the previous result is the requirement of convergence in the trace distance, whose significance  arises from Helstrom's theorem \cite{Helstrom:book, nielsen2000quantum, wilde2013quantum}. According to this theorem states with vanishing trace distance are indistinguishable and therefore equivalent resources.  
  Another new aspect of the above result is the rigorous upper bound on the achievable rate $R$.    
 Since variance is  additive for uncorrelated systems and is non-increasing in exact state conversions under TI operations, it is straightforward to show that  the  
rate $R> V_{H_1}(\psi_1)/V_{H_2}(\psi_2)$ is not achievable in  exact  state  conversions  \cite{gour2008resource}.   However, this argument fails in the presence of error $\epsilon_n$:  For a pair of output states with trace distance $\epsilon_n$, the energy variances can differ by order $\epsilon_n \lceil R n\rceil^2  \|H\|^2 $. Hence,  the variance per copy can differ by order  $ \epsilon_n  \lceil R n\rceil  \|H\|^2$, which does not necessarily vanish, even if  $\epsilon_n\rightarrow0$ in the limit $n\rightarrow\infty$.      We overcome this complication and show that with   $R> V_{H_1}(\psi_1)/V_{H_2}(\psi_2)$, error cannot vanish in the limit $n\rightarrow\infty$ (See Eq.(\ref{boundrate}) below for a more  general result).  

  Theorem \ref{Thmold} only applies   to pure states. In the rest of this paper we consider a variant of this scenario where the output  are mixed states. But, first we need to discuss the physical significance of the energy variance in this theorem.\\


\noindent\emph{Quantum Fisher Information (QFI) }-- Consider the family of states $\{e^{-i H t}\rho e^{i H t}\}_t$ corresponding to the time-evolved versions of a system in the initial state $\rho$ and Hamiltonian $H$.  The QFI relative to the time parameter $t$ for this family of state is  
\be\label{Def:Fisher}
F_H(\rho)= 2 \sum_{j,k} \frac{(p_j-p_k)^2 } {p_j+p_k}\ |\langle\phi_j| H |\phi_k\rangle|^2\ ,
\ee
where $\rho=\sum_j p_j |\phi_j\rangle\langle\phi_j|$ is the spectral decomposition of $\rho$. Equivalently, QFI can be expressed as the second derivative of the fidelity of states $\rho$ and $e^{-i H t}\rho e^{i H t}$ with respect to the parameter $t$ \cite{hayashi2006quantum}.   According to the standard interpretation of this quantity in quantum estimation,  $F_H(\rho)$ determines   how well one can estimate the unknown parameter $t$, by measuring $n\gg1$ copies of  state $e^{-i H t}\rho e^{i H t}$:  the mean squared error $\langle\delta t^2\rangle$ for any unbiased estimator satisfies the Cram\'{e}r-Rao bound $\langle\delta t^2\rangle \ge [nF_{H}(\rho)]^{-1}$, which is attainable in the asymptotic regime \cite{Holevo:book, Helstrom:book,  paris2009quantum, BraunsteinCaves:94}.   QFI has found  extensive applications in different areas of physics, beyond quantum metrology   
 (See, e.g.  \cite{girolami2013characterizing, kim2018characterizing, zhang2017detecting, girolami2017witnessing, pires2016generalized, marvian2016quantum, mondal2016quantum, zanardi2007information, zanardi2008quantum, zanardi2007bures, campos2007quantum, lashkari2016canonical}). In particular,  it has been studied  as a measure of asymmetry and coherence    \cite{yadin2016general, kwon2018clock}.  



QFI has various nice properties,  including   (\textbf{i}) Faithfulness: It is zero if, and only if, state is incoherent. (\textbf{ii}) Monotonicity: It is non-increasing under any TI operation $\mathcal{E}_\text{TI}$, i.e., $F_H(\mathcal{E}_\text{TI}(\rho))\le F_H(\rho)$. In particular, it remains invariant under energy-conserving unitaries. (\textbf{iii}) Additivity: For a composite non-interacting system with the total Hamiltonian $H_\text{tot}=H_1\otimes  I_2+I_1\otimes  H_2$, QFI is additive for uncorrelated states, i.e., $F_{H_\text{tot}}(\rho_1\otimes \rho_2)=F_{H_1}(\rho_1)+F_{H_2}(\rho_2)$. (\textbf{iv})  Convexity: For any $p\in[0,1]$ and states $\rho$ and $\sigma$, $F_H(p\rho+(1-p) \sigma)\le p F_H(\rho)+(1-p) F_H(\sigma)$.



For pure states, QFI reduces to the energy variance,  namely      $F_H(\psi)=4 V_H(\psi)$. 
 Therefore, theorem \ref{Thmold} means that  in the iid regime, the maximal rate of conversion between pure states is determined by the ratio of their QFI's.  
 This interpretation suggests that to generalize the result to mixed states, the role of variance should be replaced by QFI.  As we show  below, this conjecture  is partially correct, namely when the output  states are mixed but the input states are still pure.  On the other hand, the result of \cite{marvian2020coherence} shows that this conjecture  fails for generic mixed input states.   It is also worth noting that the state conversion   described in theorem \ref{Thmold} requires  coherent interactions between the input and output systems:  
unless the output $\psi_2$ is an energy eigenstate,  it is not possible to achieve a positive rate $R>0$ with a  vanishing error, using  measure-and-prepare (i.e., entanglement-breaking)   TI operations  \cite{marvian2020coherence}.   
This  again suggests that the operational interpretation  of QFI in the context of  parameter estimation  
cannot fully explain the special role of  variance  in theorem \ref{Thmold}.\\

\noindent\emph{Coherence cost}---Consider a system with state $\rho$ and Hamiltonian $H$ with period $\tau$. We define the coherence cost  $C^\text{TI}_c(\rho)$ of this system as the minimal rate at which c-bits with period $\tau$ (i.e., two-level systems with state $|\Theta\rangle_{\text{c-bit}}=(|0\rangle+|1\rangle)/\sqrt{2}$ and Hamiltonian $H_\text{c-bit}=\pi \sigma_z/\tau$) have to be consumed for preparing copies of this system in the iid regime, i.e., 
\be\nonumber
C^\text{TI}_c(\rho)=\inf R: \Theta_\text{c-bit}^{\otimes\lceil R n\rceil} \xrightarrow{TI}\stackrel{\epsilon_n}{ \approx}  \rho^{\otimes n}\  \text{as } n\rightarrow\infty, \epsilon_n\rightarrow 0\ , 
\ee
where the vanishing error $\epsilon_n$ is quantified in the trace distance.   This  quantity can be thought of as the counterpart of  the \emph{entanglement cost}  in  entanglement  theory \cite{horodecki2009quantum}  (Note that a different notion of coherence cost for speakable coherence is previously studied  in \cite{winter2016operational,  lami2019completing}).  Our second main result is
\begin{theorem}\label{Thm0}
The coherence cost of  a  system with  Hamiltonian $H$, state $\rho$,  and period $\tau$ is propositional to its QFI. That is  
\be
C^\text{TI}_c(\rho)=\frac{F_H(\rho)}{F_{\text{c-bit}}}= (\frac{\tau}{2\pi})^2 \times F_H(\rho) \ .
\ee
 \end{theorem}
 The lower bound $C^\text{TI}_c(\rho)\ge {F_H(\rho)}/{F_{\text{c-bit}}}$ is a special case of a more general result, which is of independent interest:  Consider a pairs of systems with states $\rho_1$ and $\rho_2$ and Hamiltonians $H_1$ and $H_2$. 
 If there exists a sequence of TI operations  converting copies of system 1 to 2 with rate $R(\rho_1\rightarrow\rho_2)$ and with a vanishing error in the trace distance (in the sense defined above), then 
  \be\label{boundrate}
R(\rho_1\rightarrow\rho_2)\le \frac{ F_{H_1}(\rho_1)}{F_{H_2}(\rho_2)}\ .
  \ee
Although this bound might be  expected from the monotonicity and additivity of QFI, as we discussed  in the case of variance,   in the presence of  
a non-zero vanishing error these properties do not necessarily imply Eq.(\ref{boundrate}).
In SM, we prove this bound using the connection between QFI and  Bures distance.   
  At the end of this Letter we sketch the proof of the other side of  theorem \ref{Thm0}.   But,  first we discuss how QFI appears in the single-copy regime. \\
%
%


\noindent\emph{From pure to mixed states in single-copy regime--}  A natural way to quantify the coherence content of a mixed state $\rho$ is to find the minimum QFI of a purification of $\rho$. More precisely,  consider an auxiliary system A  with Hamiltonian $H_A$ and let  $|\Phi_\rho\rangle_{SA}$ be a pure joint state of SA,  with the reduced state $\Tr_A(|\Phi_\rho\rangle\langle\Phi_\rho|_{SA})=\rho$.  What is the minimum possible energy variance, or, equivalently the QFI of such pure states  with respect to the total Hamiltonian  of systems S and A?  
\begin{theorem}\label{Thm:Fisher}
QFI of  system $S$ with state $\rho$ and Hamiltonian $H_S$, is four times the minimum  energy variance of all purifications of $\rho$ with auxiliary systems not interacting with $S$, i.e.
\begin{align}\label{wsingle-shot}
F_{H_S}(\rho)=\min_{\Phi_\rho, H_A}\ F_{H_\text{tot}}(\Phi_\rho)= 4\times\min_{\Phi_\rho, H_A}\ V_{H_\text{tot}}(\Phi_\rho)\ ,
\end{align}
where $H_\text{tot}=H_S\otimes I_A+I_S\otimes H_A$, and  the minimization is over all pure states $|\Phi_\rho\rangle_{SA}$ satisfying $\Tr_A(|\Phi_\rho\rangle\langle\Phi_\rho|_{SA})=\rho$, and all Hamiltonians $H_A$ of  system $A$. 
\end{theorem}
This theorem is closely related to the result of \cite{escher2011general, escher2012quantum}  in the  context of quantum metrology   (See SM for further discussion).  
  In SM we present two different  proofs of theorem \ref{Thm:Fisher}; one proof is based on  Uhlmann's theorem  \cite{nielsen2000quantum, wilde2013quantum} and the connection between fidelity and QFI  (which is similar to the argument of \cite{escher2011general}) and  the second proof is via direct minimization. The latter  approach implies that 
  for  purification $|\Phi_\rho\rangle_{SA}=\sum_j \sqrt{p_j} |\phi_j\rangle_S|\phi_j\rangle_A$ of state $\rho=\sum_jp_j|\phi_j\rangle\langle\phi_j|$ the minimum in Eq.(\ref{wsingle-shot}) is achieved for  Hamiltonian  
 \begin{align}\label{Ham}
{H_A}=-2 \sum_{j,k} \frac{\sqrt{p_j p_k} } {p_j+p_k} |\phi_j\rangle\langle\phi_k| H_S |\phi_j\rangle\langle\phi_k|\  .
\end{align}
For this Hamiltonian we find  $F_{H_S}(\rho)=4(V_{H_S}(\rho)-V_{H_A}(\rho))$,  i.e., QFI of system S is 4 times the difference between the energy variances of systems S and A. Furthermore, the QFI of  A  is  non-zero, provided that 
the QFI of S is non-zero and $\rho$ is full-rank.  This means that starting from  $|\Phi_\rho\rangle_{SA}$ by discarding A, one does not loose any QFI, even though  the discarded system itself carries non-zero QFI.  This immediately implies that the rate of distilling pure coherent states from mixed ones cannot be determined by the ratio of their QFI's. Otherwise, by distilling coherence from both systems A  
and S, we could increase 
QFI unboundedly, which is in contradiction with the general bound in Eq.(\ref{boundrate}). 
  In fact,  the distillation rate for generic mixed states is zero \cite{marvian2020coherence}.

Does this theorem determine the coherence  cost of $\rho$? From theorem \ref{Thmold} one may expect that purification $\Phi_\rho$   can be obtained by consuming c-bits at rate $({\tau}/{2\pi})^2 F_{H_\text{tot}}(\Phi_\rho)$, which in turn would imply $\rho$ can be obtained with this coherence cost. And the above theorem implies that  $F_{H_\text{tot}}(\Phi_\rho)$ can be as low as $F_{H_S}(\rho)$.    However, there is a problem with this argument: theorem \ref{Thmold} only  applies to periodic systems, whereas in general, the dynamics of   $\Phi_\rho$ under Hamiltonian $H_\text{tot}$ is not periodic.  Imposing the requirement of periodicity, in general increases the minimum variance of purification. For instance, suppose for the same purification $\Phi_\rho$ instead of Hamiltonian in Eq.(\ref{Ham}) one chooses  $H_A=-H^\ast_S$, that is the complex conjugate of  $H_S$ in the basis $\{|\phi_j\rangle\}$.  Then,  the period of  the joint system will be generally $\tau$. But,  now the energy variance is equal to $2 W_{H_S}(\rho)\ge F_{H_S}(\rho)$, where $W_{H_S}(\rho)=-\Tr([\sqrt{\rho}, H_S]^2)/2$ is the Wigner-Yanase skew information, which is another quantifier of coherence and asymmetry \cite{Marvian_thesis, marvian2014extending,  girolami2014observable, takagi2019skew}.

To overcome this issue, instead of purification, we use a different approach for preparing  $\rho$: we consider ensemble of pure states with density operator $
\rho$.  Interestingly, it turns out that there exists an optimal ensemble for which the average QFI is equal to the QFI of state $\rho$. 
\begin{theorem}\label{Thm:Petz}
QFI is four times the \emph{convex roof} of  variance. That is 
\be\label{Eq125}
F_H(\rho)= \min_{\{q_k,\eta_k\}}\sum_k q_k F_H(\eta_k)= 4\times \min_{\{q_k,\eta_k\}}\sum_k q_k V_H(\eta_k)\ ,
\ee
where the minimization is over  all ensembles of pure states $\{q_k,\eta_k\}$ satisfying $\sum_k q_k |\eta_k\rangle\langle\eta_k|=\rho$.  Furthermore, assuming the dynamics of $\rho$ under $H$ is periodic, then  the optimal ensemble can be chosen such that each $\eta_k$ is either an  eigenstate of Hamiltonian $H$ or its period under $H$ is an integer fraction of the period of $\rho$ under $H$.   
\end{theorem}
In analogy with the entanglement theory, the right-hand side of Eq.(\ref{Eq125}) can be called \emph{coherence of formation}  \cite{toloui2011constructing}.  The first part of this theorem was originally conjectured by Toth and Petz \cite{toth2013extremal} and was   later proven by Yu \cite{yu2013quantum}. Since then this result has found various applications in quantum metrology (See, e.g. \cite{bromley2017there}). Note that the convexity of $F_H$ immediately implies that if  $\sum_k q_k |\eta_k\rangle\langle\eta_k|=\rho$ then  $F_H(\rho)\le \sum_k q_k F_H(\eta_k)$, and the achievability of this bound is established in   \cite{yu2013quantum}.    Our resource-theoretic approach reveals a simple and more intuitive  proof of this fundamental property of QFI, via theorem \ref{Thm:Fisher}: Let $|\Phi_\rho\rangle_{SA}$ and $H_A$ be, respectively,  an optimal purification of  $\rho$, and the corresponding Hamiltonian of the auxiliary  system A satisfying Eq.(\ref{wsingle-shot}).  Let 
$\{|E_k\rangle\}$ be an eigenbasis of Hamiltonian $H_A$. By measuring system A in this basis,  one obtains the average joint state   $\sigma_{SA}=\sum_k  q_k \ |\eta_k\rangle\langle\eta_k|_S\otimes |E_k\rangle\langle E_k|_A$,  where $q_k$ is the probability of observing  
 $|E_k\rangle$   and  $
|\eta_k\rangle_S =\langle E_k|\Phi\rangle_{SA}/\sqrt{q_k}$ is the corresponding state of S. Then,
\be\label{bounds244}
F_{H_S}(\rho) \le F_{H_\text{tot}}(\sigma_{SA})\le F_{H_\text{tot}}(\Phi_\rho)\ .
\ee
Here, both bounds follow from the monotonicity of QFI under TI operations: State $\rho$ of system S can be obtained from $\sigma_{SA}$ by discarding system  A,  and $\sigma_{SA}$ is obtained from  $\Phi_\rho$, by measuring A in the energy eigenbasis; both operations are clearly TI.    Then, the fact that  $F_{H_\text{tot}}(\Phi_\rho)=F_{H_S}(\rho)$, implies that both bounds hold as equality. Finally, since energy eigenstates $\{|E_k\rangle\}$ have zero QFI and are orthogonal, QFI of $\sigma_{SA}$  is equal to the expected QFI of the ensemble $\{q_k, |\eta_k\rangle\}$, i.e.,  $\sum_k q_k F_{H_S}(\eta_k)=F_{H_\text{tot}}(\sigma_{SA})=F_{H_S}(\rho)$. Thus, Eq.(\ref{Eq125}) holds with   
 $|\eta_k\rangle=(\sum_j  U_{kj} \sqrt{p_j}  |\phi_j\rangle)/\sqrt{q_k}$, and  probability $ q_k=\langle E_k|\rho|E_k\rangle=\sum_j p_j |U_{kj}|^2$, where $U_{kj}=\langle E_k|\phi_j\rangle$ are the matrix elements of the unitary that diagonalizes  $H_A$ in Eq.(\ref{Ham}) in the eigenbasis of $\rho$ (Interestingly, this is the ensemble found by Yu \cite{yu2013quantum}). In summary,   the fact that QFI is  the minimum variance of purifications (Theorem \ref{Thm:Fisher}) implies that  QFI is also the  convex roof of variance (Theorem \ref{Thm:Petz}).  
The second part of theorem \ref{Thm:Petz} is shown in SM.   \\

\noindent\emph{Sketch of Proof of Theorem \ref{Thm0}}-- By combining theorems \ref{Thmold} and  \ref{Thm:Petz} with the standard typicality arguments (e.g., in \cite{hayden2001asymptotic, winter2016operational}), we show that the coherence cost of any state is determined by its QFI.  Let $(q_k, |\eta_k\rangle): k\in \mathbb{S}$ be the optimal ensemble satisfying Eq.(\ref{Eq125}). As we saw in the above proof, $\mathbb{S}$ is a finite set.     
Then, $\rho^{\otimes m}=\sum_{\textbf{k} }  q_{\textbf{k} } |\eta_\textbf{k}\rangle\langle\eta_\textbf{k}|$,  
where $\textbf{k}=k_1\cdots k_m$,    $q_{\textbf{k}}=q_{k_1}\cdots q_{k_m}$  and $|\eta_\textbf{k}\rangle=|\eta_{k_1}\rangle\cdots |\eta_{k_m}\rangle$.
  For any $k\in \mathbb{S}$ let $n_l(\textbf{k})$ be the number of occurrence of state $|\eta_l\rangle$ in $|\eta_\textbf{k}\rangle$.  Then, for $\delta >0$  define the set of typical strings as those for which the relative frequency of any  $l\in \mathbb{S}$ is between $q_l-\delta$ and $q_l+\delta$, i.e., 
 $\{\textbf{k}=k_1\cdots  k_m|\  \forall l\in \mathbb{S}:\ |\frac{n_l(\textbf{k})}{m}-q_l|\le \delta \}$. Then, 
\be
\rho^{\otimes m}=\sum_{\textbf{k}\in\text{typical} }  q_{\textbf{k} } |\eta_\textbf{k}\rangle\langle\eta_\textbf{k}|+\sum_{\textbf{k}\notin \text{typical}}  q_{\textbf{k} } |\eta_\textbf{k}\rangle\langle\eta_\textbf{k}|\  .
 \ee
 Based on this decomposition, we define a sequence of TI operations that prepare state $\rho^{\otimes m}$ with a vanishing error as $m\rightarrow\infty$: We sample string $\textbf{k}$ with probability $q_{\textbf{k} }$. If $\textbf{k}$ is not a typical string, we prepare a fixed incoherent state, which does not consume any c-bits.  By the law of large numbers, as $m\rightarrow \infty$ the probability of such events goes to zero and therefore this introduces a vanishing error.  For typical $\textbf{k}$, up to a permutation,  $|\eta_\textbf{k}\rangle$  can be written as $\bigotimes_l |\eta_l\rangle^{\otimes n_l(\textbf{k})}$, and typicality implies $n_l(\textbf{k}) \le m (q_l+\delta)$.  Therefore, $|\eta_\textbf{k}\rangle$ can be obtained from   $\bigotimes_l |\eta_l\rangle^{\otimes \lceil m (q_l+\delta) \rceil }$, which has the  energy variance $\sum_l \lceil m (q_l+\delta)\rceil  V_H(\eta_l)$.  Using the second part of theorem \ref{Thm:Petz}, one can show that  the period of this state is equal to $\tau$, the period of $\rho$.   Then, using a simple variant of theorem \ref{Thmold} we show that as  $m\rightarrow\infty$, by consuming  $({\tau}/{\pi})^2 \sum_l \lceil m (q_l+\delta)\rceil V_H(\eta_l)$ c-bits, we can prepare state  $|\eta_\textbf{k}\rangle$ with a vanishing error  (Note that the energy variance of  c-bit is $\pi^2/\tau^2$).  Using the facts that $\sum_l q_l V_H(\eta_l)=F_H(\rho)/4$ and $V_H(\eta_l)\le \|H\|^2$, where $\|H\|$ is the operator norm,  we conclude  that for any $\delta>0$,  
by consuming c-bits at rate $({\tau}/{2\pi})^2   \times (F_H(\rho)+4\delta \|H\|^2) $ per copy, one can prepare copies of the desired system with vanishing error. This proves one direction of  theorem \ref{Thm0}. See SM for further details and the proof of the other direction.  \\

\noindent\emph{Conclusion}--
In summary, preparing a general state requires consumption of both work and coherence. When coherence is a free resource, the work cost is  determined by the free energy of the system, and when work is a free resource the coherence cost is determined by QFI.  In a more complete picture both of these resources should be taken into account, and this can lead to a tradeoff between  the resources costs.  Understanding this tradeoff remains an  open question. Also, generalizing the present results to the case of non-Abelian groups, such as SO(3) will be interesting (See, e.g., \cite{yang2017units,  alexander2021infinitesimal} for progress in this direction).  Our resource-theoretic approach enabled us to clarify a previously unnoticed relation between  fundamental properties of QFI, which is arguably the most studied quantity in quantum metrology and estimation theory.  As QFI has found  extensive applications in different areas of physics,  it will be interesting to explore possible implications of  theorems  \ref{Thm0} and \ref{Thm:Fisher}  in these areas. \\

\noindent\emph{Acknowledgments}-- I am grateful to Gilad Gour and David Jennings for reading an earlier version of the manuscript carefully, and providing many useful comments and suggestions.  
Also, I would like to thank Gerardo Adesso, Anna Jen\v cov\' a,  Keiji Matsumoto,  
and Mil\'an Mosonyi for helpful discussions.   This work was supported by   NSF FET-1910571 and  NSF Phy-2046195. 
 
 \newpage

\bibliography{Ref_2018, Ref_2018_Copy, Ref_2019}



\newpage

\color{black}

\newpage

\onecolumngrid

\onecolumngrid

\newpage

\maketitle
\vspace{-5in}
\begin{center}

\Large{Supplementary Material:\\ $ $ \\   Operational Interpretation of Quantum Fisher information in Quantum Thermodynamics  }
\end{center}
\appendix

	
	\title{\textbf{Supplementary Material}:\\ $ $ \\   Coherence distillation machines are impossible in quantum thermodynamics}
	
	\author{Iman Marvian}
\affiliation{Departments of Physics \& Electrical and Computer Engineering, Duke University, Durham, North Carolina 27708, USA}



\maketitle

\color{black}

$$  $$

{\Large{\textbf{Contents}}}
$$   $$
\begin{itemize}
\item \textbf{Section \ref{Sec:app1}: Pure state transformations in the asymptotic (iid) regime}\\
In this section we prove the first part of theorem \ref{Thmold} in the paper, which determines the rate of interconversion between pure states.

\item \textbf{Section  \ref{Sec:Fisher}: Quantum Fisher Information: Preliminaries}\\

\item \textbf{Section \ref{Sec:QFIsingle}: Quantum Fisher Information as the minimum variance of purification}\\
In this section we prove theorem \ref{Thm:Fisher} in the paper, and present a new proof of theorem \ref{Thm:Petz}.

\item \textbf{Section \ref{Sec:Fisher-mon}: Monotonicity of Fisher information in  the iid regime  }\\
In this section we prove that Fisher information cannot increase in the  iid regime (Eq.(\ref{boundrate}) in the paper). This implies that Quantum Fisher information is a lower bound on coherence cost. It also proves the second part of theorem \ref{Thmold} in the paper.

\item \textbf{Section \ref{Sec:QFI:iid}: Quantum Fisher Information as the coherence cost: iid regime  }\\
In this section we prove prove theorem \ref{Thm0} in the paper which states that the coherence cost is equal to Quantum Fisher Information.

\end{itemize}

\newpage

\section{Pure state transformations in the asymptotic (iid) regime (Proof of the first part of Theorem \ref{Thmold})}\label{Sec:app1} 

In this section we study pure state conversions in the many-copy (iid) regime and prove the first part of theorem  \ref{Thmold}, namely the fact that  the interconversion between pure states is possible with a rate less than or equal to the ratio of the energy variances of the input to the output. Then, in Appendix \ref{Sec:Fisher-mon} we prove the second part of this theorem, which  implies that the  conversion is not possible with a higher rate.  

\subsection{Review of single-copy pure state transformations}\label{Sec:review}
We start by reviewing a few useful results about the single-copy pure state to pure state conversions  under TI operations.  A fundamental fact about such conversions is that the only relevant property of a pure state is its energy distribution. Let 
\be
H= \sum_E E\  \Pi_{E} 
\ee
 be the spectral decomposition of the system Hamiltonian $H$, 
where $\Pi_E$ is the projector to the subspace with energy $E$. Consider a state $\psi$ with periodic time evolution under Hamiltonian $H$ with period $\tau$, such that 
\be 
\tau=\inf_{t} \{t>0: |\langle\psi|e^{-i H t}|\psi\rangle|=1\}\ .
 \ee 
 This means that the set of energy levels $E$ with nonzero probability,  i.e., the set $\{E: \langle\psi|\Pi_{E}|\psi\rangle \neq 0\}$, can be written as
 \be\label{en-n}
 E=  n \frac{2\pi}{\tau}+E_0\  \ \ \ \ \ \ \ \ \ n\in \mathbb{Z} \ ,
 \ee 
for a fixed energy $ E_0$, satisfying
\be
0\le  E_0< \frac{2\pi}{\tau}\ .  
\ee
Then, we can describe the energy distribution of state $|\psi\rangle$ relative to Hamiltonian $H$, by a probability distribution $p_\psi$ over integer numbers $\mathbb{Z}$, defined by  
\begin{align}\label{weight}
p_\psi(n)\equiv \int_0^{2\pi} \frac{d\theta}{2\pi} \  \exp[{i (\frac{E_0 \tau}{2\pi} +n) \theta }] \times  \langle\psi| \exp({-i H  \frac{ \theta \tau}{2\pi} }) |\psi\rangle\ .
\end{align}
If $E_0+2\pi n/\tau $ is an eigenvalue of $H$, then $p_\psi(n)=\langle\psi|\Pi_{E_0+2\pi n/\tau}|\psi\rangle$ is the probability that state $|\psi\rangle$ has energy $E_0+2\pi n/\tau$, and  $p_\psi(n)=0$, otherwise.   In the following, we sometimes refer to $p_\psi$  as the energy distribution of $\psi$. 

Consider two different pure states $\psi$ and $\phi$ of a system with Hamiltonian $H$.
It can be  shown \cite{gour2008resource, marvian2013theory, marvian2014asymmetry} that  $\psi$  can be transformed to $\phi$ via an energy-conserving unitary $V$ such that  $V|\psi\rangle=|\phi\rangle$ and $[V,H]=0$, if and only if   they have the same energy distributions, i.e., 
\be
\forall n\in\mathbb{Z}:\ \ \ \  p_\psi(n)=p_\phi(n)\ ,
\ee
or equivalently, if and only if,  they have the same \emph{characteristic functions} \cite{marvian2013theory, marvian2014asymmetry}, i.e.
\be
\forall \theta\in(0,2\pi]:\ \ \ \  \langle\psi| e^{-i H\tau \frac{\theta}{2\pi} } |\psi\rangle=\langle\phi|e^{-i H\tau \frac{\theta}{2\pi}} |\phi\rangle\ .
\ee
Therefore, for a given Hamiltonian $H$,   the probability distribution $p_\psi$ specifies all the relevant information about state $\psi$ from the point of view of state conversion  under TI operations.

\begin{remark}\label{rem21}\textbf{(Shifting energy levels by a constant)} 
It is worth noting that from the point of view of state interconversions 
under TI operations, a system with Hamiltonian $H$ and a system with the Hamiltonian $H-E_0 I$, which is the shifted version of $H$, have exactly the same properties. This is because adding or subtracting a constant energy to all energy levels, is a TI operation (i.e., can be realized without interaction with a   synchronized clock). Hence, in the following discussion we always assume  
$E_0$ in Eq.(\ref{en-n}) is zero. That is, for a system with period $\tau$ the energy levels with non-zero probability are all in the form of an integer times $2\pi/\tau$, $E=  n 2\pi/{\tau}$ for $n\in\mathbb{Z}$. This condition is equivalent to $e^{-i \tau H}|\psi\rangle=|\psi\rangle$. 
\end{remark}



The above result can be generalized to the case of approximate state conversions:  if the energy distributions $p_\psi$ and  $p_\phi$ are close to each other in the \emph{total variation distance} (trace distance) then there exists a unitary transformation that converts $\psi$ to a state close to $\phi$ \cite{marvian2013theory}. In particular, there exists a unitary $V$ which commuting with the system Hamiltonian $H$ such that 
\begin{align}
|\langle\phi|V|\psi\rangle|=\sum_n \sqrt{p_\psi(n)p_\phi(n)} &\ge 1-\frac{1}{2} \|p_\psi-p_\phi\|_1\ , \\ &\equiv 1- d_{\text{TV}}(p_\psi, p_\phi)\ ,
\end{align}
where
\be
d_{\text{TV}}(p_\psi, p_\phi)\equiv \frac{1}{2} \|p_\psi-p_\phi\|_1= \frac{1}{2} \sum_n |p_\psi(n)-p_\phi(n)|\ ,
\ee
is the total variation distance (trace distance) between the two distributions (Theorem 3 in \cite{marvian2013theory}). Here, the bound follows from the Fuchs-van de Graaf inequality $1-\sqrt{\text{Fid}(\rho,\sigma)}\le \frac{1}{2}\|\rho-\sigma\|_1$.

In terms of trace distance, this means that there exists an energy-conserving unitary $V$ such that
\be
\frac{1}{2}\|V|\psi\rangle\langle\psi|V^\dag-|\phi\rangle\langle\phi| \|_1=\sqrt{1-|\langle\phi|V|\psi\rangle|^2} =\sqrt{1- \big[\sum_n \sqrt{p_\psi(n)p_\phi(n)}\big]^2}\le \sqrt{2 d_\text{TV}(p_\psi,p_\phi)}\ ,
\ee
where the bound again follows from the Fuchs-van de Graaf inequality $1-\sqrt{\text{Fid}(\rho,\sigma)}\le \frac{1}{2}\|\rho-\sigma\|_1$, which implies $1-{\text{Fid}(\rho,\sigma)}\le \|\rho-\sigma\|_1\ . $

In summary, if the energy distributions are close in the total variation distance, then the pair of states can be converted to each other, with good approximation. 

In addition to the energy-conserving unitaries, TI operations also include operations  that do not conserve the system energy. In particular, consider two systems  with Hamiltonian $H_1$ and $H_2$ and states $\psi_1$ and $\psi_2$, respectively, and assume both states  have period $\tau$.  Let $p_{\psi_1}$ and $p_{\psi_2}$ be the energy distributions for these two states defined via Eq.(\ref{weight}). 
Then, these states are interconvertable to each other via TI operations if, and only if  there exists an integer $k$ such that 
\be\label{eq:en}
\forall n\in\mathbb{Z}:\ \ \ \ \ \ p_{\psi_1}(n)=p_{\psi_2}(n+k)\ .
\ee
 In the special case where the input and output Hamiltonians are identical, i.e., $H_1=H_2$  this operation adds  energy $-k2\pi/\tau$ to the system.  Again, if Eq.(\ref{eq:en}) holds approximately, then the conversion can be implemented approximately via a TI operation, with an error  determined by the total variation distance between the probability distribution $p_{\psi_1}(n)$ and $p_{\psi_2}(n+k)$. 

The following proposition summarizes these results

\begin{proposition}\label{prop:TI}\textbf{(based on \cite{gour2008resource, marvian2013theory, marvian2014asymmetry} )}\label{prop}
Suppose two systems  with Hamiltonian $H_1$ and $H_2$ and states $\psi_1$ and $\psi_2$ both  have period $\tau$.  Let $p_{\psi_1}$ and $p_{\psi_2}$ be, respectively,  the energy distributions for pure state $\psi_{1}$ and $\psi_{2}$, defined  in Eq.(\ref{weight}).   Then, for any integer $k$, there exists a TI operation $\mathcal{E}_{TI}$ such that $\mathcal{E}_{TI}(|\psi_1\rangle\langle\psi_1|)$ is a pure state that satisfies
\be
\langle\psi_2|\mathcal{E}_{TI}(|\psi_1\rangle\langle\psi_1|)|\psi_2\rangle= \Big(\sum_{n\in\mathbb{Z}}  \sqrt{p_{\psi_1}(n) p_{\psi_2}(n+k)}\Big)^2\ ,
\ee
and
\begin{align}
\frac{1}{2}\Big\| \mathcal{E}_{TI}(|\psi_1\rangle\langle\psi_1|) - |\psi_2\rangle\langle\psi_2| \Big\|_1 &=\sqrt{1-\Big[\sum_n  \sqrt{p_{\psi_1}(n) p_{\psi_2}(n+k)}\Big]^2}\\ &\le  \sqrt{\sum_n \big|p_{\psi_1}(n)-p_{\psi_2}(n+k)\big|}\ .
\end{align}
\end{proposition}

\color{black}

  \subsection{State conversions in the iid regime}

Next, we consider the iid regime: Suppose we are given $m$ copies of a system with Hamiltonian $H$ and state $\psi$, i.e.,  non-interacting systems with the joint state  $\psi^{\otimes m}$ and the total (non-interacting) Hamiltonian  $H_\text{tot}=\sum_{i=1}^m H^{(i)} $, where $H^{(i)}=I^{\otimes (i-1)}\otimes H\otimes I^{\otimes (m-i)}$.  The total energy for these systems is the sum of the energy of the individuals. Therefore, for state $\psi^{\otimes m}$ the probability distribution over  energy eigenspaces of $H_\text{tot}$ is equal to the probability distribution for the random variable $n_\text{tot}=n_1+\cdots  +n_m$, where each integer random variables $n_k$ has the  probability distribution $p_\psi$. Hence, the probability distribution of the total energy for state $\psi^{\otimes m}$ is given by the $m$-fold convolution of the probability distribution $p_\psi$, i.e.
\be\label{mfold}
 p_{\psi^{\otimes m}}=\underbrace{p_{\psi}\ast \cdots \ast p_{\psi}}_\text{$m$ times}\ .
 \ee

Ref. \cite{schuch2004nonlocal, schuch2004quantum, gour2008resource} argue that in the limit of large number of copies $m\gg1$, the central limit theorem implies that the distribution of energy for $\psi^{\otimes m}$ converges to a  Gaussian distribution. Any such distribution is determined by only two parameters, namely the variance and the mean.  It follows that if the energy variances for two states $\psi^{\otimes m}$ and $\phi^{\otimes \lceil R m\rceil}$  match approximately, then by adding or subtracting energy, which is a TI operation, we can shift the center of the distributions and overlap them. This is the main intuition in the arguments of \cite{schuch2004nonlocal, schuch2004quantum, gour2008resource}.  

Although this intuition is correct, there are some crucial details which require more careful analysis. Most importantly, the standard central limit theorems do not guarantee the convergence in the total variation distance , which requires stronger assumptions.  To prove this stronger notion of convergence, we need to use more advanced results on the limit theorems, which are reviewed  in the following section. Furthermore, the argument of \cite{schuch2004nonlocal, schuch2004quantum, gour2008resource}, is restricted to the case of states with gapless spectrum, i.e., those for which the the support of distribution $p_\psi$ is a single connected interval of integers. As we will see in the following, this assumption is not necessary.

\subsection{A local limit theorem and convergence in the total variation distance}

In this section we review a result of \cite{barbour2002total}, which shows that under certain conditions,  sum of integer-valued random variables converges to a translated Poisson distribution, in the total variation distance (See also \cite{rollin2015local}).

 In the following $Y \sim P(\sigma^2)$ means the integer-valued random variable $Y$ has Poisson distribution with variance $\sigma \ge 0$, such that  any integer $l\ge 0$ occurs with probability ${e^{-\sigma} \sigma^{l}}/{l!}$. 
 The Poisson distribution is specified by only one parameter $\sigma$, which determines both the variance and the mean of the distribution. We are  interested in the more general family of integer-valued distributions  obtained by translating Poisson distributions with integers, such that the variance and mean can be independent of each other. However, by translating with an integer we can only change the mean of the distribution in a discrete fashion.  It follows that using this family of distributions we cannot really achieve arbitrary mean and variance. Nevertheless, for any desired mean $\mu$ and variance $\sigma^2$ we can find a  translated Poisson distribution whose mean is exactly $\mu$ and its variance is close to $\sigma^2$, such that their difference is less than one. \\

 \noindent \textbf{Translated Poisson Distribution:} For any given $\mu$ and $\sigma^2>0$, let $Z\sim TP(\mu, \sigma^2)$ be a random variable which satisfies $Z-s \sim P(\sigma^2+\gamma)$  where the shift  $s:=
\lfloor\mu-\sigma\rfloor$ is an integer, and $\gamma:=\mu- \sigma^2 -\lfloor\mu-\sigma^2\rfloor$, satisfies $0 \le \gamma < 1$. This means  $Z-s$ has Poisson distribution with variance $\sigma^2+\gamma$, i.e.,    
\be\label{TP:def}
Z\sim TP(\mu, \sigma^2)\  \Longleftrightarrow\  Z-\lfloor\mu-\sigma\rfloor \sim P(\sigma^2+\gamma) \ .
\ee
It follows that the random variable $Z\sim TP(\mu, \sigma^2)$ has mean $\mu$, i.e. $\mathbb{E} Z=\mu$, and its variance is $\mathbb{E} Z^2-(\mathbb{E} Z)^2= \sigma^2+\gamma$, which is between $\sigma^2$ and $\sigma^2+1$.\\

Let 
\be
W=\sum_{i=1}^m X_i\ ,
\ee
 be the sum of $m$ independent integer-valued random variables $X_i$, with mean  $\mu_i=\mathbb{E} X_i$ and variance $\sigma_i^2=\text{Var} X_i$, and bounded third moment, i.e. $\mathbb{E} |X^3_i|< \infty $. Let $\mu=\mathbb{E}  W\equiv \sum_{i=1}^m \mu_i$, and $\sigma^2\equiv \sum_{i=1}^m \sigma^2_i$ be the variance of $W$.  
 Finally, define
\be\label{jaja2021}
\phi_i \equiv \sigma^2_i \mathbb{E}\{X_i (X_i - 1)\} + {|\mu_i - \sigma^2_i|}\ \mathbb{E}\{(X_i -1)(X_i -2)\}+\mathbb{E}|X_i(X_i - 1)(X_i - 2)|\ .
\ee
Note that if $\sigma^2_i$ and the third moment $\mathbb{E} X^3_i$ are both finite, then $\phi_i$ is also a finite number. 

Let $\mathcal{L}(X_i)$ be the distribution of the random variable $X_i$. In the following result we assume 
\be\label{cond2021}
d_{\text{TV}}(\mathcal{L}(X_i), \mathcal{L}(X_i+1))<1 , 
\ee
which means $X_i$ is not perfectly  distinguishable from its translated version. This is true if there is $n\in\mathbb{Z}$ such that $X_i$ takes both values $n$ and $n+1$, with non-zero probabilities.  

Roughly speaking, the following theorem states that if 
all random variables $X_i: i=1,\cdots m$ satisfy the condition in Eq.(\ref{cond2021}) and have  bounded third moments and  nonzero variances, then  the sum $W=\sum_i X_i$ converges to a translated Poisson distribution.
\begin{theorem}(\emph{Corollary 3.2 in Barbour-Cekanavicius}\cite{barbour2002total})\label{lem:prob}
Consider random variables $X_i: i=1,\cdots, m$ with mean  $\mu_i=\mathbb{E} X_i$ and  variance $\sigma_i^2=\text{Var} X_i$.
 Let 
\be
a=\min_{i=1,\cdots, m}  \sigma^2_i\ , \ \  \ \  \ \  \  b=\min_{i=1,\cdots, m} \nu_i \ ,\ \ \ \ \  \ \  c=\max_{i=1,\cdots, m} \ \ \frac{\phi_i}{\sigma^2_i }\ ,
\ee
where $\nu_i= \min\{\frac{1}{2}, 1-d_{\text{TV}}(\mathcal{L}(X_i), \mathcal{L}(X_i+1))  \}$, and $\phi_i$ is defined in Eq.(\ref{jaja2021}).  Assume $a, b>0$ and $c<\infty$.  Then, the total variation distance of the distribution of $W=\sum_{i=1}^m X_i $ and the translated Poisson distribution $TP(\mu,\sigma^2)$ is bounded by 
\be
d_{\text{TV}}{\large{(}}\mathcal{L}(W), TP(\mu, \sigma^2) {\large{)}}=\frac{1}{2}\|\mathcal{L}(W)-TP(\mu, \sigma^2)  \|_1  \le  \frac{c}{\sqrt{m b-1/2}}+\frac{2}{m a} \ .
\ee 
\end{theorem}  

We end this section by recalling another useful result on the total variation distance between Poisson distributions (See \cite{adell2006exact}).
 \begin{lemma} \textbf{\cite{adell2006exact}}\label{lem:dis}
 The total variation distance between two Poisson distributions with variances $\sigma^2+x$ and $\sigma^2$, for $x\ge 0,$ is bounded by
\begin{align}
d_{\text{TV}}(P(\sigma^2), P(\sigma^2+x) )&= \frac{1}{2}\|P(\sigma^2)-P(\sigma^2+x) \|_1\\ &=\frac{1}{2} \sum_n \left|\frac{e^{-\sigma} \sigma^{n}}{n!}-\frac{e^{-\sqrt{\sigma^2+x}} \sqrt{(\sigma^2+x)^{n}}}{n!}  \right|\\ 
&\le \min\{x, \sqrt{\frac{2}{e}}(\sqrt{\sigma^2+x}-\sigma) \}\ .
\end{align}
\end{lemma}
Therefore, for variance $\sigma^2> 0$, we find that the total variation distance between   $P(\sigma^2)$ and $P(\sigma^2+x) $ is bounded by the ratio of the difference between the two variances to the square root of the variance, i.e. 
\be
d_{\text{TV}}(P(\sigma^2), P(\sigma^2+x) )\le \frac{x}{\sigma}\ .
\ee

\subsection{State conversion in the iid regime  (Proof of the first part of Theorem \ref{Thmold})}\label{Sec:app:spec}



Next, we apply this result  to study the conversion of pure  states in the iid regime using TI operations. As we saw before, for $m$ copies of a system with state $\psi$ and Hamiltonian $H$, the total energy is $ 2\pi n_\text{tot}/\tau$, where 
\be
n_\text{tot}=n_1+ \cdots  +n_m\ ,
\ee
and $n_i$ has distribution $p_\psi$.  We denote the distribution of this random variable by  $p_{\psi^{\otimes m}}$, which is the  $m$-fold convolution of $p_{\psi}$, as in Eq.(\ref{mfold}).   
Applying theorem \ref{lem:prob}  we find that,  provided that certain conditions (listed below) are satisfied,  this distribution can be approximated by a translated Poisson distribution $TP(\mu,\sigma^2)$ with the mean 
\be\label{mean21}
\mu=  \mathbb{E}\{n_\text{tot}\}=  m\times \mathbb{E}\{n\}= m\times  \frac{\tau}{2\pi} \langle\psi|H|\psi\rangle \ ,
\ee
and  the variance
\begin{align}\label{variance21}
\sigma^2&=m\times (\mathbb{E}\{n^2\}-\mathbb{E}^2\{n\})=m (\frac{\tau}{2\pi})^2 \times V_H(\psi) \ ,
\end{align}
 where $V_H(\psi)=\langle\psi|H^2|\psi\rangle-\langle\psi|H|\psi\rangle^2$ is the energy variance of state $\psi$ for Hamiltonian $H$ (Recall that $\psi$ has only components on the eigen-subspaces of  $H$ with eigenvalue in the form  of an integer times $2\pi/\tau$. See remark \ref{rem21}). 
 
 In particular,  theorem \ref{lem:prob} implies 
\be\label{156tr}
d_{\text{TV}}{\large{(}}p_{\psi^{\otimes m}},TP(\mu, \sigma^2) {\large{)}}\le  \frac{c}{\sqrt{m b-1/2}}+\frac{2}{m [(\tau/2\pi)^2\times V_H(\psi)]  } \ .
\ee
Here,
\be
b= \min\big\{\frac{1}{2}, 1-\frac{1}{2}\sum_l |p_\psi(l)- p_\psi(l+1)|  \big\}\ \ ,
\ee
where 
$\frac{1}{2}\sum_l |p_\psi(l)- p_\psi(l+1)|$ is the total variation distance between $p_\psi$ and the translated version of $p_\psi$,
and $c$, defined in theorem  \ref{lem:prob} is a finite number (independent of $m$), provided that the energy variance of $\psi$ is nonzero and it has bounded third moment of energy. 

It follows that in the limit $m$ goes to infinity  the two distributions $p_{\psi^{\otimes m}}$ and $TP(\mu, \sigma^2) $ converge in the total variation distance, if the following conditions are all satisfied:
\begin{enumerate}
\item The distribution $p_\psi$ has a nonzero variance, which means $\psi$ is not an eigenstate of the system Hamiltonian $H$.
\item The distribution $p_\psi$ has a finite third moment (This also guarantees that coefficient $c$ in Eq.(\ref{156tr}) is finite).
\item  The total variation distance between $p_\psi$ and the translated version of $p_\psi$ satisfies
\be\label{cond3}
\frac{1}{2}\sum_n |p_\psi(n)- p_\psi(n+1)| < 1\ ,
\ee 
which means  the two distributions $p_\psi$ and $\tilde{p}_\psi$, defined by  $\tilde{p}_\psi(n)={p}_\psi(n+1)$,  have overlapping supports.
This condition is satisfied if there exists, at least, an integer $n_0\in \mathbb{Z}$ such that both $p_\psi(n_0)$ and $p_\psi(n_0+1)$ are non-zero.

\end{enumerate}

From the discussion in Sec.\ref{Sec:review} and, in particular, proposition \ref{prop:TI} we know that  the interconvertability of pure states under TI operations are determined by their energy distribution.   Combining this proposition  with the above result,  we can study interconversion of  systems in the iid regime.

Consider two systems with states  $\psi_1$ and $\psi_2$ and Hamiltonians $H_1$ and $H_2$, respectively. Suppose  both systems have period $\tau$. Then, assuming the above conditions are satisfied, in the limit $m$ goes to infinity the energy distribution of $\psi_1^{\otimes m}$, denoted by $ p_{\psi_1^{\otimes m}}$,  
converges to  $TP(\mu_1, \sigma_1^2) $, where $\mu_1$ and $\sigma_1$ are  the mean and energy variance of $\psi_1^{\otimes m}$, defined via Eq.(\ref{mean21}), i.e., 
\begin{align}
\mu_1&=m (\frac{\tau}{2\pi}) \times \langle\psi_1|H_1|\psi_1\rangle\ \ \ , \  \ \ \  \ \ 
\sigma_1^2=m (\frac{\tau}{2\pi})^2\times V_{H_1}(\psi_1)\ .
\end{align}

Similarly, consider  $\lceil R m\rceil$ copies of system with Hamiltonian $H_2$ and state $\psi_2$, where 
\be
R=\frac{V_{H_1}(\psi_1)}{V_{H_2}(\psi_2)}\ .
\ee
Let   $p_{\psi^{\otimes \lceil R m\rceil}_2}$ be the energy distribution for state  $\psi_2^{\otimes \lceil R m\rceil}$. Then, in the limit of large $m$ the energy distribution for state $\psi_2^{\otimes \lceil R m\rceil}$ converges to the translated Poisson distribution $TP(\mu_2,\sigma_2^2)$,  where
\begin{align}
\mu_2=\lceil R m\rceil \times (\frac{\tau}{2\pi}) \langle\psi_2|H_2|\psi_2\rangle\  \ , \  \ \ \  \ \  
\sigma_2^2&=\lceil R m\rceil \times (\frac{\tau}{2\pi})^2\times V_{H_2}(\psi_2)\ .
\end{align}  
Recall that the distribution $TP(\mu, \sigma)$ is the distribution obtained from translating a Poisson distribution with variance $\sigma^2+\gamma$ with an integer, where $0 \le \gamma\le 1$  (See Eq.(\ref{TP:def})). 
 Therefore, up to a translation by integers $TP(\mu_1, \sigma_1^2)$ and $TP(\mu_2, \sigma_2^2)$ are, respectively,  equal to the Poisson distributions  $P(\sigma_1^2+\gamma_1)$ and $P(\sigma_2^2+\gamma_2)$, where $0\le \gamma_{1,2}\le 1$.

In summary, up to translations by integers, the distributions    $ p_{\psi_1^{\otimes m}}$ and $p_{\psi^{\otimes \lceil R m\rceil}_2}$ are equal to Poisson distributions   $P(\sigma_1^2+\gamma_1)$ and $P(\sigma_2^2+\gamma_2)$, respectively, whose total variation distance is bounded by
\bes
\begin{align}
d_{\text{TV}}\big(P(\sigma_1^2+\gamma_1) , P(\sigma_2^2+\gamma_2) \big)&\le \frac{|\sigma^2_2-\sigma^2_1+\gamma_2-\gamma_1|}{\sigma_1}\\ &\le \frac{1}{\sigma_1}\big(|\sigma^2_1-\sigma^2_2|+|\gamma_1-\gamma_2|\big)
\\ 
 &\le \frac{1}{\sqrt{q\times m }}  \Big(q\times \Big|m- \frac{\lceil R\times m\rceil}{R}    \Big|+2\Big)\\ 
 &\le \frac{1}{\sqrt{q\times m }}  \Big( \frac{q}{R}+2\Big)=\frac{1}{\sqrt{ m \tau V_{H_1}(\psi_1)/2\pi   }}  \Big(\frac{\tau V_{H_2}(\psi_2)}{2\pi}+2\Big)\ ,
\end{align}
\ees
where $q=\tau V_{H_1}(\psi_1)/2\pi$ and the first bound is obtained by applying lemma \ref{lem:dis}. In conclusion, if $q=\tau V_{H_1}(\psi_1)/2\pi>0$, then in the limit $m$ goes to infinity, the total variation distance between distribution $p_{\psi^{\otimes \lceil R m\rceil}_2}$  and a properly translated version of  $ p_{\psi_1^{\otimes m}}$ goes to zero, with an error upper bounded by $(qR^{-1}+2)/\sqrt{qm}$. Note that the required amount of translation is an integer.

Next, we apply proposition  \ref{prop:TI}. According to this proposition,  if by translating with an integer we can convert the distribution  $ p_{\psi_1^{\otimes m}}$ to a distribution close to $p_{\psi^{\otimes \lceil R m\rceil}_2}$, with the total variation distance $\epsilon$, then there exists a TI operation that converts state  $\psi_1^{\otimes m}$ to state  $\psi^{\otimes \lceil R m\rceil}_2$ with trace distance $\sqrt{2\epsilon}$.  Therefore,  we arrive at 
the following result:
 \begin{proposition}\label{cor12}
Consider two systems  with Hamiltonian $H_1$ and $H_2$ and states $\psi_1$ and $\psi_2$, respectively. Assume:
\begin{enumerate}
\item Both systems  have period $\tau$, such that  
\be
\tau= \inf_{t}\{t>0: \big|\langle\psi_{l}|e^{-i H_{l} t}|\psi_{l}\rangle \big|=1\}\ :\ \ \ \ \ l=1,2\ .
\ee
\item   Suppose the systems have non-zero energy variances   $V_{H_1}(\psi_1)$ and $V_{H_2}(\psi_2)$ and their third moments of energy is finite, i.e., $|\langle \psi_l|H_l^3|\psi_l\rangle|<\infty$  for $l=1,2$.

\item  The energy distributions $p_{\psi_{1,2}}$ satisfy the condition
\be\label{cond62}
\frac{1}{2}\sum_n |p_{\psi_{l}}(n)-p_{\psi_{l}}(n+1)| < 1\ : \ \ \ \ l=1,2\ .
\ee
where $
p_{\psi_{l}}(n)= \frac{1}{2\pi} \int_0^{2\pi}d\theta\  e^{i\theta n } \langle\psi_{l}|e^{-i H_{l}\tau \frac{\theta}{2\pi}} |\psi_{l}\rangle\ $, is the probability that state $\psi_{l}$ has energy ${2\pi n/\tau}$ with respect to Hamiltonian $H_{l}$, where we have defined the energy references for Hamiltonians $H_1$ and $H_2$ such that $e^{-i H_{l} \tau}|\psi_{l}\rangle=|\psi_{l}\rangle\ .$ 
 \end{enumerate}
Let  $R=\frac{V_H(\psi_1)}{V_H(\psi_2)}$ be the ratio of energy variances.   Then, for any integer $m$ there exists a TI operation $\mathcal{E}_m$ that maps $\psi_1^{\otimes m}$ to a state close to $\psi_2^{\otimes \lceil R m\rceil}$, such that their trace distance vanishes in the limit $m$ goes to infinity, i.e.,
\be\label{rgrg}
\lim_{m\rightarrow \infty} \|\mathcal{E}_m(\psi_1^{\otimes m})-\psi_2^{\otimes \lceil R m\rceil}\|_1=0\ .
\ee
 
\end{proposition}

It turns out that the last condition   in the above proposition, i.e.,  Eq.(\ref{cond62}) is not necessary.  
 We explain this with the following example.   Consider the energy distributions associated to states 
$$|\eta\rangle=\frac{|0\rangle+|2\rangle}{\sqrt{2}}\ \ , \ \  \ \  \text{and}\ \ \ \ \ \ \ \ |\gamma\rangle=\frac{|0\rangle+|2\rangle+|5\rangle}{\sqrt{3}}\ , $$
with the Hamiltonian $H=2\pi\tau^{-1} \sum_{k=0}^\infty k |k\rangle\langle k| $. One can easily see that although neither the distribution $p_\eta$ nor the distribution $p_\gamma$ do not satisfy the condition in Eq.(\ref{cond62}), there is an important distinction between them:
Suppose instead of one copy of state $|\gamma\rangle$ we look at the energy distribution for two copies of this state, which is given by the distribution $p_{\gamma^{\otimes 2}}=p_\gamma\ast p_\gamma$. This distribution has  support on $n=0, 2, 4, 5, 7, 10$. It follows that, even though the energy distribution for one copy of $\gamma$ does not satisfy Eq.(\ref{cond62}), energy distribution for two copies of this state \emph{does}   
satisfy this condition. That is  the total variation distance between $p_\gamma\ast p_\gamma(n)$ and its translated version $p_\gamma\ast p_\gamma(n+1)$ is less than one,
\be
\frac{1}{2} \sum_n |p_\gamma\ast p_\gamma(n+1)-p_\gamma\ast p_\gamma(n) | <1\ .
\ee

 Thus, we can apply the above result  to two copies of this state and conclude that, in the limit $m$ goes to infinity,  the energy distribution for $(\gamma^{\otimes 2})^{\otimes m}$ converges to   a translated Poisson distribution.

On the other hand, this will not happen  for state $|\eta\rangle$: Since the support of $p_\eta$ is restricted to even integers $n=0,2$,  for any integer $L$,  the support of  $p_{\eta^{\otimes L}}$ is also restricted to even integers. Therefore, in the limit of large $L$, the energy distribution will not converge to a translated Poisson distribution (In fact, it converges to a translated Poisson distribution defined only on even integers).

It turns out that the distinction between these two examples have a simple physical interpretation, in terms of the period of dynamics. Recall that the period of dynamics for a system with state $\psi$ and Hamiltonian $H$ is defined as $ \inf_{t}\{t>0: \big|\langle\psi|e^{-i H t}|\psi\rangle \big|=1\}\ $. It can be easily seen that for state  $|\eta\rangle$ the period of dynamics is $\tau/2$, whereas  for state $|\gamma\rangle$ the period is  $\tau$. Using the following lemma, we can show that, in general, having the full period $\tau$ is the necessary and sufficient condition to guarantee that condition in Eq.(\ref{cond62}) is satisfied for a finite number of copies of state. 

\begin{lemma}\label{lem:app:54}
Consider a system with  Hamiltonian $H$, state $|\psi\rangle$ and period $\tau= \inf_{t}\{t>0: \big|\langle\psi|e^{-i H t}|\psi\rangle \big|=1\}$. Let  $p_{\psi}(n)$ defined in Eq.(\ref{weight}) 
be the probability that state $\psi$ has energy $2\pi n/\tau$ (Recall that we assume $E_0=0$, which can be always achieved by a proper shift of the Hamiltonian). 
Then, there is a finite $L$ such that the distribution $p_{\psi^{\otimes L}}=\underbrace{p_{\psi}\ast \cdots \ast p_{\psi}}_\text{$L$ times}$, corresponding to the energy distribution of $\psi^{\otimes L}$, satisfies
\be\label{ref156}
\frac{1}{2} \sum_{n\in\mathbb{Z}} | p_{\psi^{\otimes L}}(n)-p_{\psi^{\otimes L}}(n+1) | <1\ .
\ee
 \end{lemma}
We prove this lemma at the end of this section, using  Bezout's theorem.


In conclusion, if the system has period $\tau$, then there exists a finite positive integer $L$, such that $\psi^{\otimes L}$ satisfies the condition in Eq.(\ref{cond62})  of   proposition  \ref{cor12}. Therefore, we can apply this proposition to state $\psi^{\otimes L}$. Note that the energy variance of this state $L\times V_H(\psi)$, and if $\psi$ has finite third moment, then so does  $\psi^{\otimes L}$. In summary, we find 

\begin{theorem}\label{Thm:main:app}
Consider a pair of systems  with Hamiltonians $H_1$ and $H_2$ and states $\psi_1$ and $\psi_2$, respectively. Assume both systems have period  $\tau$,  and finite non-zero energy variances $V_{H_1}(\psi_1),V_{H_2}(\psi_2) >0$, and finite third moments of energy.   
 Let $R\le  \frac{V_H(\psi_1)}{V_H(\psi_2)}$. Then, for any integer $m$ there exists a TI operation  $\mathcal{E}_m$ that converts $\psi_1^{\otimes m}$ to a state close to $\psi_2^{\otimes \lceil R m\rceil}$, such that their trace distance vanishes in the limit $m$ goes to infinity, as stated in Eq.(\ref{rgrg}.)  
 \end{theorem}
In Section \ref{Sec:Fisher-mon}, theorem \ref{Thm:mono}, we prove a converse bound, that is we show  that the state conversion  is impossible with a vanishing error with rate  $R>\frac{V_H(\psi_1)}{V_H(\psi_2)}$. In theorem \ref{Thm:mono}, the result is presented in terms of QFI. Note that for pure states, the energy variance is one fourth of QFI, i.e. $F_H(\psi)=4V_H(\psi)$. We finish this section by proving lemma \ref{lem:app:54}. 
 
 \subsection*{Proof of lemma \ref{lem:app:54}}
 
Let  $n_\text{min} 2\pi/\tau$ be the minimum occupied energy level by state $\psi$ (Note that any Hamiltonian has a lowest energy level). In other words, let
\be
n_\text{min}=\min \{n: p_\psi(n)\neq 0\} \ ,
\ee
be the minimum $n$ for which $p_\psi(n)\neq 0$. Let
\be
\mathcal{N}_\psi=\{n-n_\text{min}: p_\psi(n)\neq 0\}
\ee
be the set of all occupied levels shifted by $n_\text{min}$. The fact that the period is $\tau$ implies that the greatest common divisor of this set is 1, i.e.
\be
\text{gcd}(\mathcal{N}_\psi)=1\ .
\ee
This can be seen by noting that if $k=\text{gcd}(\mathcal{N}_\psi)$, then for any $n$ either $p_\psi(n)=0$ or $n-n_\text{min}= jk $ for an integer $j$. Therefore, since energy levels are related to integer $n$ via relation   $E=2\pi n/\tau$, we find
\be
|\langle\psi|e^{-i H\tau/k} |\psi\rangle|=|\sum_n p_\psi (n) e^{- i 2\pi (jk+n_\text{min})/k}|=|\sum_n p_\psi (n) e^{- i 2\pi n_\text{min}/k}|=1\ ,
\ee
which implies the period  is smaller than $\tau$. Therefore, assuming the period is $\tau$, we have $\text{gcd}(\mathcal{N}_\psi)=1$.

Next, we use Bezout's theorem:
\begin{lemma} \emph{(Bezout's theorem)}
Suppose the greatest common divisor of a set integers $\{a_1,\cdots, a_n\}$ is one, i.e. $\text{gcd} (\{a_1,\cdots, a_n\})=1$.  Then, there exists integers $\{x_1,\cdots, x_n\}$, such that $\sum_{i=1}^n x_i a_i=1$.
\end{lemma}
We apply this result to the set of integers $\{n_i\}_i=\mathcal{N}_\psi$. Then, the fact that the greatest common divisor of this set is one implies that there exists a set of integers $\{x_i\}_i$ such that
\be
\sum_i x_i n_i = 1\ .
\ee
Partitioning the set $\{x_1,\cdots, x_n\}$ to two subsets which only include  positive and negative elements of  this set,  we find
\be
\sum_{i: x_i>0} x_i n_i = 1-\sum_{i: x_i<0} x_i n_i =1+\sum_{i: x_i<0} |x_i| n_i  \ .
\ee
Let $L=\sum_i |x_i|$ and consider the probability distribution 
$$p_{\psi^{\otimes L}}=\underbrace{p_{\psi}\ast \cdots \ast p_{\psi}}_\text{$L$ times}$$
 corresponding to the total energy distribution for state $\psi^{\otimes L}$. This is the probability distribution for the random variable $\sum_{r=1}^L Z_r$, assuming each $Z_r$ has the distribution $p_\psi$. We show that for this distribution 
 \be
 \frac{1}{2} \sum_{n\in\mathbb{Z}} | p_{\psi^{\otimes L}}(n)-p_{\psi^{\otimes L}}(n+1) | <1 .
 \ee
 
To show this we argue that  the random variable $\sum_{r=1}^L Z_r$ takes both values $K$ and $K-1$, with a non-zero probability, where 
 \be
K=n_\text{min} L+\sum_{i: x_i>0} x_i n_i=n_\text{min} L-\sum_{i: x_i<0} x_i n_i +1\ .
\ee
To see this, first consider the following event: For each $n_i\in \mathcal{N}_\psi$ with $x_i>0$, suppose  $x_i$ different random variables in the set $\{Z_r: 1\le r\le L \}$  take the value $n_\text{min}+n_i$, and the rest of the random variables, i.e. $L-\sum_{i: x_i>0} x_i$ random variables, take the value $n_\text{min}$. In this event,  the sum  $\sum_{r=1}^L Z_r$ will be equal to $K=n_\text{min} L+\sum_{i: x_i>0} x_i n_i $. It follows that 
\be
p_{\psi^{\otimes L}}(K)> 0\ .
\ee
Next, consider a different event in which for each $x_i<0$, $|x_i|$ different random variables in the set $\{Z_r: 1\le r\le L \}$ take the value $n_\text{min}+n_i$, and the rest of the random variables in this set, i.e. $L-\sum_{i: x_i<0} |x_i|$, take the value $n_\text{min}$. In this event  the sum  $\sum_{r=1}^L Z_r$ will be equal to $K-1=n_\text{min} L+\sum_{i: x_i<0} |x_i| n_i $. It follows that 
\be
p_{\psi^{\otimes L}}(K-1)> 0 \ .
\ee
We conclude that the distribution 
$p_{\psi^{\otimes L}}=\underbrace{p_{\psi}\ast \cdots \ast p_{\psi}}_\text{$L$ times}$
 is nonzero for  both $K$ and $K-1$. This immediately implies  
\be
\frac{1}{2}\sum_n |p_{\psi^{\otimes L}}(n)- p_{\psi^{\otimes L}}(n+1) | <1\ ,
\ee
and proves lemma \ref{lem:app:54}.

 \newpage
\section{Quantum Fisher Information: Preliminaries}\label{Sec:Fisher}
Here, we briefly review some useful properties of Quantum Fisher Information (QFI). See e.g.  \cite{Holevo:book, Helstrom:book, petz2011introduction, paris2009quantum, BraunsteinCaves:94}  for further details.

QFI for a general family of states   $\rho_t$ labeled by the real continuous parameter $t$ is defined by
\be\label{Fisher}
I_F(t)=\Tr(\rho_t L_t^2)\ ,
\ee
where $L_t$ is the \emph{symmetric logarithmic derivative}, defined via equation
\be
\dot{\rho}_t=\frac{1}{2}(\rho_t L_t+L_t\rho_t)\ .
\ee
In the special case of  $\rho_t=e^{-i t H}\rho e^{i t H}$ for a Hermitian operator $H$, we find
\be\label{Fish-app11}
\dot{\rho}(t)= -i [H,\rho_t]=\frac{1}{2}(\rho_t L_t+L_t\rho_t)\ .
\ee
Using    the spectral decomposition of state $\rho$, as $\rho=\sum_k p_k |\phi_k \rangle\langle\phi_k|$ we find
\be
2i\times  \frac{p_k-p_j}{p_k+p_j} \langle\phi_k|L_t |\phi_j\rangle =\langle\phi_k| e^{i H t} L_t e^{-i H t} |\phi_j\rangle\ .
\ee
Putting this back into Eq.(\ref{Fisher}) we find
\begin{align}
I_F(t)&=\Tr({\rho}_t L_t^2)=\Tr({\rho} L_0^2)=I_F ,
\end{align}
i.e., the QFI is independent of the parameter $t$, and therefore we denote it by $I_F$. Then, it can be easily seen that
\bes
\begin{align}
I_F=I_F(t)&=\Tr({\rho}_t L_t^2)\\ &=\sum_{k,j} p_k |\langle\phi_k|L_0|\phi_j\rangle|^2\\ &=4\sum_{k,j} p_k  \frac{(p_k-p_j)^2}{(p_k+p_j)^2} |\langle\phi_k|H|\phi_j\rangle|^2\\ &=2\sum_{k,j} (p_k+p_j)  \frac{(p_k-p_j)^2}{(p_k+p_j)^2} |\langle\phi_k|H|\phi_j\rangle|^2\\ &=2\sum_{k,j}   \frac{(p_k-p_j)^2}{p_k+p_j} |\langle\phi_k|H|\phi_j\rangle|^2\ .
\end{align}
\ees
Note that if $\rho$ is not full rank,   we can apply the above formula to the state $\rho_\epsilon=(1-\epsilon)\rho+ \epsilon I/d$ for a vanishing $\epsilon\rightarrow 0$,  where $I/d$ is the totally mixed state. 
Using this technique, or  applying the definition in Eq.(\ref{Fish-app11})  we find that for pure states QFI is four time the variance of state $\psi$ with the respect to the observable $H$, i.e.
\be
I_F=4\times ( \langle\psi|H^2|\psi\rangle- \langle\psi|H|\psi\rangle^2)=4V_H(\psi)\ .
\ee
In the following, we use the notation $F_H(\rho)$ to denote QFI for the family of state $e^{-i H t}\rho e^{i H t}: t\in\mathbb{R}$. In summary,   for a system with state $\rho$ and Hamiltonian $H$, QFI is given by
\be
F_H(\rho)=2\sum_{i,j}   \frac{(p_k-p_j)^2}{p_k+p_j} |\langle\phi_k|H|\phi_j\rangle|^2\ .
\ee
QFI is closely related to fidelity. Let 
\be
\text{Fid}(\rho,\sigma)\equiv\Tr\Big(\sqrt{\sqrt{\rho}\sigma\sqrt{\rho} }\Big)^2=\|\sqrt{\rho}\sqrt{\sigma}\|^2_1 \ ,
\ee
be the fidelity of states $\rho$ and $\sigma$. Consider the fidelity of 
state $\rho$ and  $e^{-i H t}\rho e^{i H t}$ as a function of $t$. For $t=0$ fidelity takes its maximum value, which is one. Therefore, its first derivative with respect to $t$ vanishes, i.e.
\be
\frac{d}{dt}\  \text{Fid}(\rho, e^{-i H t}\rho e^{i H t})\Bigr|_{t=0}\ =0 .
\ee
Furthermore, it turns out that the second derivative is given by QFI, i.e. 
\be\label{conec}
F_H(\rho) =-4 \frac{d^2}{dt^2}\  \sqrt{\text{Fid}(\rho, e^{-i H t}\rho e^{i H t})}\Bigr|_{t=0}\ .
\ee
Therefore, roughly speaking, QFI determines how fast states $\rho$ and $e^{-i H t}\rho e^{i H t}$ become distinguishable.

QFI has the following important properties \cite{Holevo:book, Helstrom:book, petz2011introduction, paris2009quantum, BraunsteinCaves:94}:

\begin{enumerate}
\item \textbf{Faithfulness}: It is zero if, and only if, state is incoherent, i.e., diagonal  in the energy eigenbasis. This can be seen using the fact that  $[\rho,H]=0$ if and only if  for all $i$ and $j$,
\be
\langle\phi_i| [\rho,H]|\phi_j\rangle =(p_i-p_j) \langle\phi_i|H|\phi_j\rangle=0\ ,
\ee  
where $\rho=\sum_j p_j |\phi_j\rangle\langle\phi_j|$ is the spectral decomposition of $\rho$. Using the formula
\be
F_H(\rho)=2\sum_{i,j}   \frac{\Big[(p_i-p_j) |\langle\phi_i|H|\phi_j\rangle|\Big]^2}{p_i+p_j} \ ,
\ee
we can easily see that this is the case if, and only if, $F_H(\rho)= 0$.

\item  \textbf{Monotonicity}: It is non-increasing under any TI operation $\mathcal{E}_\text{TI}$, i.e. 
\be
F_H(\mathcal{E}_\text{TI}(\rho))\le F_H(\rho) .
\ee
 In particular, it remains invariant under energy-conserving unitaries. This can be easily seen, e.g., using the connection between QFI and the fidelity, and the fact that fidelity satisfies information processing inequality, i.e.
\be
\text{Fid}(\rho,\sigma)\le \text{Fid}(\mathcal{E}(\rho),\mathcal{E}(\sigma)) \ ,
\ee
for any trace-preserving completely positive map $\mathcal{E}$.

\item \textbf{Additivity}: For a composite non-interacting system with the total Hamiltonian $H_\text{tot}=H_1\otimes  I_2+I_1\otimes  H_2$, QFI is additive for uncorrelated states, i.e. $F_{H_\text{tot}}(\rho_1\otimes \rho_2)=F_{H_1}(\rho_1)+F_{H_2}(\rho_2)$.  This can be seen, e.g.,  from the multiplicativity of the fidelity for tensor products,  together with the connection between fidelity and QFI in Eq.(\ref{conec}).

\item \textbf{Convexity}: For any $0 \le p\le 1$ and states $\rho$ and $\sigma$, $F_H(p\rho+(1-p) \sigma)\le p F_H(\rho)+(1-p) F_H(\sigma)$.  This also can be seen from the concavity  of the fidelity  together with the connection between fidelity and QFI in Eq.(\ref{conec}).
\end{enumerate}

\newpage

\section{Quantum Fisher Information in the single-shot regime}\label{Sec:QFIsingle}
In this section we prove theorem \ref{Thm:Fisher} in the paper. For completeness we repeat the statement of this theorem. \\

\noindent \textbf{Theorem}\ (Restatement of theorem \ref{Thm:Fisher}) \emph{QFI of  system $S$ with state $\rho$ and Hamiltonian $H_S$, is four times the minimum  energy variance of all purifications of $\rho$ with auxiliary closed systems not interacting with $S$, i.e.
\begin{align}\label{single-shot}
F_{H_S}(\rho)=\min_{\Phi_\rho, H_A}\ F_{H_\text{tot}}(\Phi_\rho)= 4\times\min_{\Phi_\rho, H_A}\ V_{H_\text{tot}}(\Phi_\rho)\ ,
\end{align}
where the minimization is over all pure states $|\Phi_\rho\rangle_{SA}$ satisfying $\Tr_A(|\Phi_\rho\rangle\langle\Phi_\rho|_{SA})=\rho$, and all Hamiltonians of the purifying system $A$.
}\\

\noindent\textbf{Previous Result:} We note that a closely related result has been previously obtained in the context of quantum metrology \cite{escher2011general} (See also \cite{escher2012quantum}). This reference 
considers a general family of states $\rho_x$ of  system S and shows that there exists a purification $|\Psi_x\rangle_{SA}$ of this system, such that the QFI of $\rho_x$ is equal to QFI of $|\Psi_x\rangle_{SA}$. Applying this
result together with the Uhlmann's theorem for the family of states $\rho_t=e^{-i H_S t}\rho e^{i H_S t}$, one finds that there is a purification of $\rho$, denoted by $|\Psi\rangle_{SA}$, and a family of unitaries $U_A(t)$ on the purifying system A such that the QFI of the family of states $\rho_t=e^{-i H t}\rho e^{i H t}$ is  equal to the QFI of the family of states 
$[e^{-i H_S t}\otimes U_A(t)]|\Psi_x\rangle_{SA}$ \cite{escher2012quantum}.  However, note that this argument does not imply that $U_A(t)$ itself can be written as  $e^{-i H_A t}$ for a time-independent Hamiltonian $H_A$. In particular, the family of unitaries $U_A(t)$  found by applying the Uhlmann's theorem is not necessarily differentiable and, in fact, is not unique for  degenerate $\rho$ . \\

In the following we present two different proofs of this theorem, one proof is via direct minimization and the second proof, similar to the approach of \cite{escher2011general},  is based on the connection of fidelity and QFI together with Uhlmann's theorem. However, in the second proof we assume $\rho$ is non-degenerate and full-rank. \\

\subsection{First proof of theorem \ref{Thm:Fisher} via direct minimization}

Consider system $S$ with Hamiltonian $H_S$ and state $\rho$ with the spectral decomposition $\rho=\sum_i {p_i} |\phi_i\rangle\langle \phi_i|$. Consider an auxiliary system $A$ with Hamiltonian $H_A$. Let  $|\Phi_\rho\rangle$ be a pure state of systems $A$ and $S$ which purifies state $\rho_S$, such that 
\be
\rho_S=\Tr_A(|\Phi_\rho\rangle\langle\Phi_\rho|) ,
\ee
where the partial trace is over system $A$. 

Let $H_\text{tot}$ be the total Hamiltonian of the system $S$ and auxiliary  system $A$, i.e.
\be
H_\text{tot}=H_S\otimes I_A+I_S\otimes H_A \ .
\ee
We are interested in finding the purification $|\Phi_\rho\rangle$ and  Hamiltonian $H_A$ for which the total energy variance   
\be
V_{H_\text{tot}}(|\Phi_\rho\rangle)= \langle \Phi_\rho|H^2_\text{tot}|\Phi_\rho\rangle-\langle \Phi_\rho|H_\text{tot}|\Phi_\rho\rangle^2
\ee
is minimized.  Since all purifications of $\rho$ are equal up to a unitary on system $A$ we can fix the purification to be 
\be
|\Phi_\rho\rangle=\sum_i \sqrt{p_i} |\phi_i\rangle|\phi_i\rangle=(\sqrt{\rho}\otimes I) \sum_i  |\phi_i\rangle|\phi_i\rangle\ ,
\ee
and only vary the Hamiltonian $H_A$. For this purification the reduced state on system $A$ is also state $\rho$, i.e.
\be
\Tr_S(|\Phi_\rho\rangle\langle\Phi_\rho|)=\rho .
\ee
Next, note that by adding a proper multiple of the identity operator to $H_A$, we can always make the expectation value of the total energy zero, such that 
\be
\langle \Phi_\rho|H_\text{tot} |\Phi_\rho\rangle=0 .
\ee
But, adding a multiple of the identity operator to the Hamiltonian  does not change the energy variance. Therefore, in the following, without loss of generality, we assume the expectation value of the total Hamiltonian $H_\text{tot} $ is zero. This means that the energy variance is given by the following expectation value 
\begin{align}
V_{H_\text{tot}}(|\Phi_\rho\rangle)&= \langle \Phi_\rho|H^2_\text{tot}|\Phi_\rho\rangle\\ &=\langle \Phi_\rho|H^2_S\otimes I_A |\Phi_\rho\rangle+\langle \Phi_\rho| I_S\otimes H^2_A |\Phi_\rho\rangle+2 \langle \Phi_\rho| H_S\otimes H_A |\Phi_\rho\rangle\ .
\end{align}
Then, using $|\Phi_\rho\rangle=(\sqrt{\rho}\otimes I) \sum_i  |\phi_i\rangle|\phi_i\rangle$ we find
\bes\label{Var}
\begin{align}
V_{H_\text{tot}}(|\Phi_\rho\rangle)&=\Tr(\rho H^2_S)+\Tr(\rho {H^2_A})+2\Tr(\sqrt{\rho} H_S \sqrt{\rho} H^T_A)\\ &=\Tr(\rho H^2_S)+\Tr(\rho ({H^T_A})^2)+2\Tr(\sqrt{\rho} H_S \sqrt{\rho} H^T_A)\ ,
\end{align}
\ees
where $T$ denotes the transpose relative to the eigenbasis of $\rho$, i.e. $\{|\phi_j\rangle\}_j$. Here, to get the second line we have used 
$\Tr(\rho ({H^T_A})^2)=\Tr(\rho {H^2_A})$, which follows from the fact that  the trace of any operator remains invariant under transpose, together with the fact that $\rho$ is diagonal in   $\{|\phi_j\rangle\}_j$ basis, and so $\rho^T=\rho$.

Next, we consider small variations of $H^T_A$, denoted by $\delta H^T_A$. At the point where the variance $V_{H_\text{tot}}(|\Phi_\rho\rangle)$  is minimized, we have  
\begin{align}
\frac{\delta V_{H_\text{tot}}(|\Phi_\rho\rangle)}{\delta H^T_A}  =\frac{\delta \langle \Phi_\rho|H^2_\text{tot}|\Phi_\rho\rangle}{\delta H^T_A} =0\ .
\end{align}
Using Eq.(\ref{Var}) it can be easily seen that
\bes
\begin{align}
\delta V_{H_\text{tot}}(|\Phi_\rho\rangle)&=\delta \langle \Phi_\rho|H^2_\text{tot}|\Phi_\rho\rangle\\ &=\Big[\Tr( \rho (\delta H^T_A) {H^T_A}  )+\Tr(\rho  {H^T_A}\delta H^T_A)+2\Tr(\sqrt{\rho} H_S \sqrt{\rho} \delta H^T_A)\Big]\\ &\ \ \ \ +\mathcal{O}((\delta H^T_A)^2) .
\end{align}
\ees
At the point where the variance is minimized,  this variation vanishes up to the first order with respect to $\delta H^T_A$, for all variations $\delta H^T_A$. This leads to the equation
\begin{align}\label{Eq13}
\frac{{H^T_A} \rho+\rho  {H^T_A}}{2}=-\sqrt{\rho} H_S \sqrt{\rho} \ ,
\end{align}
which should be satisfied by $H^T_A$ for which the variance is minimized. 

Next, we  find  ${H^T_A}$ which satisfies this equation. To solve this equation we vectorize both side, using the relation
\be
Y=\sum_{i,j} Y_{i,j} |\phi_i\rangle\langle \phi_j| \longleftrightarrow \text{vec}(Y)=\sum_{i,j} Y_{i,j} |\phi_i\rangle|\phi_j\rangle\ ,
\ee
which implies 
\be\label{vec}
\text{vec}(XYZ)=(X\otimes Z^T)\text{vec}(Y) .
\ee
Using this notation we can rewrite Eq.(\ref{Eq13}) as
\begin{align}
[I\otimes \rho^T+\rho\otimes I] \text{vec}({H^T_A})=-2  [\sqrt{\rho}\otimes \sqrt{\rho}^T] \text{vec}({H_S})\ .
\end{align}
This equation implies
\begin{align}
\text{vec}({H^T_A})=-2  [I\otimes \rho^T+\rho\otimes I]^{-1} [\sqrt{\rho}\otimes \sqrt{\rho}^T] \text{vec}({H_S})\ .
\end{align}
Using the decomposition $\rho=\sum_{i} p_i |\phi_i\rangle\langle\phi_i|$ we find
\begin{align}
\text{vec}({H^T_A})&=-2  [I\otimes \rho^T+\rho\otimes I]^{-1} [\sqrt{\rho}\otimes \sqrt{\rho}^T] \text{vec}({H_S})\\
 &=-2  \Big[\sum_{i,j} (p_i+p_j) |\phi_i\rangle\langle\phi_i|\otimes|\phi_j\rangle\langle\phi_j|  \Big]^{-1} [\sqrt{\rho}\otimes \sqrt{\rho}] \text{vec}({H_S})
 \\
 &=-2  \Big[\sum_{i,j} (p_i+p_j)^{-1} |\phi_i\rangle\langle\phi_i|\otimes|\phi_j\rangle\langle\phi_j|  \Big] [\sqrt{\rho}\otimes \sqrt{\rho}] \text{vec}({H_S})
 \\
 &=-2  \Big[\sum_{i,j} \frac{\sqrt{p_i p_j} } {p_i+p_j} |\phi_i\rangle\langle\phi_i|\otimes|\phi_j\rangle\langle\phi_j|  \Big] \text{vec}({H_S})\ .
\end{align}
Using Eq.(\ref{vec}) this implies
\begin{align}
{H^T_A}=-2 \sum_{i,j} \frac{\sqrt{p_i p_j} } {p_i+p_j} |\phi_i\rangle\langle\phi_i| H_S |\phi_j\rangle\langle\phi_j|  \ ,
\end{align}
or, equivalently, 
\begin{align}
{H_A}=-2 \sum_{i,j} \frac{\sqrt{p_i p_j} } {p_i+p_j} |\phi_j\rangle\langle\phi_i| H_S |\phi_j\rangle\langle\phi_i|  . 
\end{align}
Note that 
\bes
\begin{align}
\Tr(\rho{H_A})&=-2 \sum_{i,j} \frac{\sqrt{p_i p_j} } {p_i+p_j} \Tr(\rho|\phi_j\rangle\langle\phi_i| H_S |\phi_j\rangle\langle\phi_i|)\\ &=- \sum_{i} p_i \langle\phi_i| H_S |\phi_i\rangle\\ &=-\Tr(\rho H_S)  . 
\end{align}
\ees
It follows that the expectation value of the total Hamiltonian is zero, i.e. $ \langle \Phi_\rho|H_\text{tot} |\Phi_\rho\rangle=0 .$

For this optimal $H_A$ we have
\bes\label{Eq09}
\begin{align}
\Tr(\rho{H^2_A}) &=4  \Tr(\rho\Big[\sum_{i,j} \frac{\sqrt{p_i p_j} } {p_i+p_j} |\phi_j\rangle\langle\phi_i| H_S |\phi_j\rangle\langle\phi_i| \Big]\Big[ \sum_{k,l} \frac{\sqrt{p_k p_l} } {p_k+p_l} |\phi_l\rangle\langle\phi_k| H_S |\phi_l\rangle\langle\phi_k| \Big])\\ &=4 \sum_{i,j} \frac{{p_i p^2_j} } {(p_i+p_j)^2} |\langle\phi_i| H_S |\phi_j\rangle|^2\\ &=2 \sum_{i,j} \frac{{p_i p^2_j}+{p_j p^2_i} } {(p_i+p_j)^2} |\langle\phi_i| H_S |\phi_j\rangle|^2\\ &=2 \sum_{i,j} \frac{{p_i p_j}} {p_i+p_j} |\langle\phi_i| H_S |\phi_j\rangle|^2\ ,
\end{align}
\ees
where to get the third line we have used the fact that $\frac{1} {(p_i+p_j)^2} |\langle\phi_i| H_S |\phi_j\rangle|^2$ is symmetric with respect to $i$ and $j$.

Similarly, 
\bes
\begin{align}
\Tr(\sqrt{\rho} {H_S} \sqrt{\rho} {H^T_A})&=-2  \Tr\Big( {H_S} \sqrt{\rho}   \sum_{i,j} \frac{\sqrt{p_i p_j} } {p_i+p_j} |\phi_i\rangle\langle\phi_i| H_S |\phi_j\rangle\langle\phi_j| \sqrt{\rho} \Big)\\ &=-2  \Tr\Big( {H_S}  \sum_{i,j} \frac{{p_i p_j} } {p_i+p_j} |\phi_i\rangle\langle\phi_i| H_S |\phi_j\rangle\langle\phi_j|  \Big)\\ &=-2 \sum_{i,j} \frac{{p_i p_j} } {p_i+p_j} |\langle\phi_i| H_S |\phi_j\rangle|^2\ \\ &=-\Tr(\rho{H^2_A})\ ,
\end{align}
\ees
where to get the last line we have used Eq.(\ref{Eq09}).

Putting these into Eq.(\ref{Var})  we find
\bes
\begin{align}
V_{H_\text{tot}}(|\Phi_\rho\rangle)&=\Tr(\rho H^2_S)+\Tr(\rho {H^2_A})+2\Tr(\sqrt{\rho} H_S \sqrt{\rho} H^T_A)\\ &=\Tr(\rho H^2_S)-\Tr(\rho {H^2_A}) \\ &= \Tr(\rho H^2_S)-2 \sum_{i,j} \frac{{p_i p_j} } {p_i+p_j} |\langle\phi_i| H_S |\phi_j\rangle|^2
\\ &= \sum_i p_i  \langle\phi_i|H^2_S|\phi_i\rangle-2 \sum_{i,j} \frac{{p_i p_j} } {p_i+p_j} |\langle\phi_i| H_S |\phi_j\rangle|^2
\\ &= \sum_{i,j} p_i  |\langle\phi_i|H_S|\phi_j\rangle|^2-2 \sum_{i,j} \frac{{p_i p_j} } {p_i+p_j} |\langle\phi_i| H_S |\phi_j\rangle|^2
\\ &= \sum_{i,j} \frac{(p_i+p_j)^2}{2(p_i+p_j)}  |\langle\phi_i|H_S|\phi_j\rangle|^2-2 \sum_{i,j} \frac{{p_i p_j} } {p_i+p_j} |\langle\phi_i| H_S |\phi_j\rangle|^2
\\ &= \sum_{i,j} \frac{(p_i-p_j)^2}{2(p_i+p_j)}  |\langle\phi_i|H_S|\phi_j\rangle|^2\ , 
\end{align}
\ees
where to get the fifth line we have used the decomposition of the identity operator as $\sum_j |\phi_j\rangle\langle\phi_j|$, and to get the sixth line we have used the fact that $ |\langle\phi_i|H_S|\phi_j\rangle|^2$ is symmetric with respect to $i$ and $j$.

Comparing this with the formula for QFI 
\be
F_H(\rho)=2\sum_{i,j}   \frac{(p_k-p_j)^2}{p_k+p_j} |\langle\phi_k|H|\phi_j\rangle|^2\ ,
\ee
we find that
\begin{align}
V_{H_\text{tot}}(|\Phi_\rho\rangle)&=\frac{1}{4} F_H(\rho)\ . 
\end{align}
This completes the proof.

It is also worth noting that because $\Tr(H_S \rho)=-\Tr(H_A \rho)$ and $V_{H_\text{tot}}(|\Phi_\rho\rangle)=\Tr(\rho H^2_S)-\Tr(\rho {H^2_A}) $, we have
\be
\frac{1}{4} F_H(\rho)=V_{H_\text{tot}}(|\Phi_\rho\rangle)=\Tr(\rho H^2_S)-\Tr(\rho {H^2_A})=V_{H_S}(\rho)-V_{H_A}(\rho)\ ,
\ee
i.e., QFI is four times the difference between energy variance of system S and the auxiliary system A. 

\subsection*{{Fisher information of the purifying system}}
The above argument shows that if  the Hamiltonian of the auxiliary system is 
\begin{align}
{H_A}=-2 \sum_{i,j} \frac{\sqrt{p_i p_j} } {p_i+p_j} |\phi_j\rangle\langle\phi_i| H_S |\phi_j\rangle\langle\phi_i|  , 
\end{align}
then for the total Hamiltonian $H_S\otimes I_A+I_S\otimes H_A$ of the composite system  $S$ and $A$, the QFI of state $|\Phi_\rho\rangle=\sum_i \sqrt{p_i} |\phi_i\rangle|\phi_i\rangle\ $, is equal to the QFI for system $S$. In other words, by discarding system $A$ the QFI does not decrease. It is interesting to note that this happens even though the QFI of the auxiliary system $A$ is nonzero.

To  calculate the QFI of the auxiliary system, first note that the reduced state of system $A$ in this case is also $\rho$. Then, using the formula for QFI we find   
\begin{align}
F_{H_A}(\rho)&= 2\sum_{i,j} \frac{(p_i-p_j)^2}{(p_i+p_j)}  |\langle\phi_i|H_A|\phi_j\rangle|^2\\ 
&= \sum_{i,j} \frac{2(p_i-p_j)^2}{(p_i+p_j)}  \frac{4 p_i p_j } {(p_i+p_j)^2}  |\langle\phi_i|H_S|\phi_j\rangle|^2\\ &= \sum_{i,j} \frac{8 p_i p_j (p_i-p_j)^2}{(p_i+p_j)^3}   |\langle\phi_i|H_S|\phi_j\rangle|^2 
\end{align}
Therefore, if the system $S$ is in  a full rank density operator with nonzero Fisher information, then the Fisher information for the auxiliary system will be necessarily nonzero, $F_{H_A}(\rho)>0$. 

We conclude that for state $|\Phi_\rho\rangle=\sum_i \sqrt{p_i} |\phi_i\rangle|\phi_i\rangle\ $, and  for this choice of Hamiltonian $H_A$,  by discarding system $A$, the Fisher information   does not decrease, even though the process is irreversible, and the discarded system itself carries non-zero Fisher information.

\subsection*{Comparison with the Wigner-Yanase skew Information}
In the above argument we found the optimal Hamiltonian of auxiliary system for the joint state  $|\Phi_\rho\rangle=\sum_i \sqrt{p_i} |\phi_i\rangle|\phi_i\rangle\ $. Since for this joint state the reduced state of both subsystems $A$ and $S$ is $\rho$, a natural choice for the Hamiltonian $H_A$  which minimizes the total energy variance could be $H_A=-H^T_S=-H^\ast_S$, where $T$ denotes the transpose relative to the eigenbasis of $\rho$ and $\ast$ denotes complex conjugation in this basis. Then, the total energy variance is given by
\begin{align}
V_{H_\text{tot}}(|\Phi_\rho\rangle)&= \langle \Phi_\rho|H^2_\text{tot}|\Phi_\rho\rangle-\langle \Phi_\rho|H_\text{tot}|\Phi_\rho\rangle^2\\ &=\langle \Phi_\rho|H^2_S\otimes I_A |\Phi_\rho\rangle+\langle \Phi_\rho| I_S\otimes H^2_A |\Phi_\rho\rangle+2 \langle \Phi_\rho| H_S\otimes H_A |\Phi_\rho\rangle\\ &=\Tr(\rho H^2_S)+\Tr(\rho {H^2_A})+2\Tr(\sqrt{\rho} H_S \sqrt{\rho} H^T_A)\\ &=2\Tr(\rho H^2_S)-2\Tr(\sqrt{\rho} H_S \sqrt{\rho} H_S) ,
\end{align}
where, in the first line we have used the fact that for $H_A=-H^T_S$ the expectation value of total energy is zero. 

Interestingly, the last line is twice the Wigner-Yanase skew information 
\be
W_H(\rho)=-\frac{1}{2}\Tr\big([H_S,\sqrt{\rho}]^2\big)\ ,
\ee
which is also  a measure of asymmetry relative to time-translations. Therefore, for this choice of $H_A$ we find
\begin{align}
V_{H_\text{tot}}(|\Phi_\rho\rangle)=2 W_H(\rho) .
\end{align}

\subsection{Second Proof of theorem \ref{Thm:Fisher} via Uhlmann's theorem and Differentiability  of Singular Value Decomposition}
Here, we prove the result under the extra assumption that $\rho$ is non-degenerate and full-rank.  

We are interested to find a purification $|\Phi_\rho\rangle_{SA}$ of state $\rho$ and a Hamiltonian $H_A$ acting on the purifying system $A$, such that
\be
V_{H_\text{tot}}(|\Phi_\rho\rangle)= \langle \Phi_\rho|H^2_\text{tot}|\Phi_\rho\rangle-\langle \Phi_\rho|H_\text{tot}|\Phi_\rho\rangle^2
\ee
is minimized, where $H_\text{tot}=H_S\otimes I_A+I_S\otimes H_A$. First, using the fact that QFI is non-increasing under partial trace, and the QFI for pure states is 4 times the variance, we find that
\be\label{upper-jul}
\frac{1}{4} F_{H_S}(\rho) \le  V_{H_\text{tot}}(|\Phi_\rho\rangle)\ .
\ee
Next, we prove that there exists a Hamiltonian $H_A$ for which this holds as an equality.   To prove this we use the connection between fidelity and QFI.  Recall that for a system with state $\sigma$ and Hamiltonian $H$, the QFI is equal to
\be\label{Fish-jul}
F_H(\sigma)=-4  \frac{d^2}{dt^2} \sqrt{\text{Fid}(\sigma, e^{-i H t}\sigma e^{i H t})}\Bigr|_{t=0}\ .
\ee
Then, we use the following fact which is proven later:
 \begin{lemma}\label{lem-item1}
 Let $\rho$ be a full-rank density operator with non-degenerate eigenvalues, pure state $|\Phi_\rho\rangle_{SA}$ be a purification of $\rho$, and $H_S$ be a bounded Hamiltonian. Then, there exists a family of unitary operators $V_A(t)$, satisfying 
 \be\label{item1}
 \sqrt{\text{Fid}(\rho,e^{-i H_S t}\rho e^{i H_S t})}= \langle\Phi_\rho|[e^{-i H_S t}\otimes V_A(t) ]  |\Phi_\rho\rangle_{SA}\ ,
 \ee
 for all $t$, such that (i)  $V_A(0)$ is the identity operator, and (ii)  $V_A(t)$ is infinitely  differentiable in a finite neighborhood around $t=0$\ .
\end{lemma} 
Then, applying this lemma we find
\bes\label{tht}
\begin{align}
\frac{d^2}{dt^2}\sqrt{\text{Fid}(\rho,e^{-i H_S t}\rho e^{i H_S t})}\Big|_{t=0} &=  \frac{d^2}{dt^2} \Big|\langle\Phi_\rho|[e^{-i H_S t}\otimes V_A(t) ]  |\Phi_\rho\rangle\Big|_{t=0}\\ &=-V_{H_\text{tot}}(\Phi_\rho)
 \  ,
 \end{align}
 \ees
 where $H_\text{tot}=H_S\otimes I_A+I_S\otimes H_A$ and 
\be
H_A= i \frac{d}{dt}V_A(t)\Big|_{t=0}\ .
\ee
Here, to get the second line in Eq.(\ref{tht}) one can use the identity
\be
\Big(\frac{d^2}{dx^2} \Big|\langle\psi_0|\psi_x\rangle\Big|\Big)_{x=0}=\Big| \frac{d\langle\psi_x|}{dx}|\psi_x\rangle\Big|_{x=0}^2- \frac{d\langle\psi_x|}{dx}\frac{d|\psi_x\rangle}{dx}\Big|_{x=0} \ ,
\ee
which holds for any smooth family of states $|\psi_x\rangle$. 

Using Eq.(\ref{Fish-jul}), this in turn implies
\be
F_{H_S}(\rho)=4 \times V_{H_\text{tot}}(|\Phi_\rho\rangle)\ .
 \ee
Combining this with inequality (\ref{upper-jul}),  we find that
\be
F_{H_S}(\rho) = 4\times \min_{H_A} V_{H_\text{tot}}(|\Phi_\rho\rangle)
\ee 
Therefore, to complete the proof of theorem \ref{Thm:Fisher}  we need to prove lemma \ref{lem-item1}.

\subsection*{Smooth purifications (Proof of lemma \ref{lem-item1})}
Let $\rho=\sum_j p_j |\phi_j\rangle\langle\phi_j|$ be the spectral decomposition of $\rho$, and 
\be
|\Phi_\rho\rangle_{SA} =(\sqrt{\rho}\otimes I_A)\sum_j |\phi_j\rangle_S |\phi_j\rangle_A=\sum_j \sqrt{p_j} |\phi_j\rangle_S |\phi_j\rangle_A\  . 
\ee
Then, any purification of $e^{-i H t}\rho e^{i H t}$ can be written as $[e^{-iH_S t}\otimes V_A(t) ]|\Phi_\rho\rangle_{SA}$ for a unitary $V_A(t)$.  According to the Uhlmann's theorem, there exists a unitary $V_A(t)$ such that  
\be
\langle\Phi_\rho|[e^{-iH_S t}\otimes V_A(t) ]|\Phi_\rho\rangle_{SA} =  \sqrt{\text{Fid}(\rho, e^{-iH_S t}{\rho} e^{iH_S t})} =\|\sqrt{\rho} e^{-iHt} \sqrt{\rho} e^{iH_St} \|_1=\|\sqrt{\rho} e^{-iHt} \sqrt{\rho}  \|_1 \ ,
\ee
which means 
\be\label{aha}
\Tr\Big(\sqrt{\rho} e^{-i H_S t} \sqrt{\rho}\  V_A^T(t)\Big)=\big\|\sqrt{\rho} e^{-i H_S t} \sqrt{\rho}\big\|_1\ .
\ee
In fact, $V_A^T(t)$ can be determined directly from the singular value decomposition of $\sqrt{\rho} e^{-i H_S t} \sqrt{\rho}$. Consider the decomposition  
\be\label{SVD}
\sqrt{\rho}e^{-iH_S t} \sqrt{{\rho}} =L(t) D(t) R(t)\ ,
\ee
where $D(t)$ is a diagonal matrix with non-negative elements and $L(t)$ and $R(t)$ are unitary transformations.  Then, in Eq.(\ref{aha}) we can choose
\be
V^T_A(t)=[L(t)R(t)]^\dag=R^\dag(t) L^\dag(t)\ .
\ee
In general, singular value decomposition is not necessarily smooth.  The following result, guarantees  the smoothness of this decomposition under certain conditions:
\begin{theorem}(\cite{dieci1999smooth})
Let $A(t)$ be $k$-times continuously differentiable square matrix function with full-rank and distinct singular values. Then,  $A(t)$ has a Singular Value Decomposition which is $k$-times continuously differentiable.  
\end{theorem}
We apply this theorem to operator $\sqrt{\rho} e^{-i H t} \sqrt{\rho}$ and its singular value decomposition in Eq.(\ref{SVD}). Note that this operator is infinitely differentiable. Furthermore,  if $\rho$ is full-rank, then $\sqrt{\rho} e^{-i H t} \sqrt{\rho}$ is also full-rank. And, if $\rho$ is non-degenerate then for sufficiently small $t$ the singular values of $\sqrt{\rho} e^{-i H t} \sqrt{\rho}$ will  be distinct (Note that the singular values of this operator are square root of the eigenvalues of operator $\sqrt{\rho} (e^{-i H t} {\rho} e^{i H t})\sqrt{\rho} $. For sufficiently small $t$, these eigenvalues will be arbitrary  close to the eigenvalues of $\rho^2$, which are distinct). 

Therefore, applying the above result we conclude that if $\rho$ is full-rank and its eigenvalues are distinct, then the operator $V^T_A(t)=[L(t)R(t)]^\dag=R^\dag(t) L^\dag(t)$ is infinitely differentiable at $t=0$, which in turn implies $V_A(t)$ is infinitely differentiable at $t=0$ and completes the proof.

\subsection{QFI as the convex roof of variance (Proof of theorem \ref{Thm:Petz})}

For completeness we restate theorem \ref{Thm:Petz} in more details. \\

\noindent\textbf{Restatement of theorem 4:} 
Suppose under Hamiltonian $H$ state $\rho$ has period $\tau$, i.e. $
\tau=\inf_t\{t>0: e^{-i H t} \rho e^{i H t}=\rho\} $\ . Then,  
\be
F_H(\rho)= \min_{\{q_k,\eta_k\}}\sum_k q_k F_H(\eta_k)= 4\times \min_{\{q_k,\eta_k\}}\sum_k q_k V_H(\eta_k)\ ,
\ee
where the minimization is over all ensembles $\{q_k,|\eta_k\rangle\}$ satisfying $\sum_{k} q_k |\eta_k\rangle\langle\eta_k|=\rho$. Furthermore, the optimal ensemble $\{q_k,|\eta_k\rangle\}$ for which the minimum is achieved can be chosen such that   the period   $\tau_k=\inf_t\{t>0:  |\langle\eta_k| e^{-i H t} |\eta_k\rangle|=1 \}\ $ of state $|\eta_k\rangle$ under Hamiltonian $H$ is either $0$, i.e., $\eta_k$ is an eigenstate of Hamiltonian $H$, or $\tau_k=\tau/m_k$ for an integer $m_k\in \mathbb{N}$.  Finally, at each time $t\in (0,\tau)$ there is  at least one state $ |\eta_k\rangle$ with non-zero probability $q_k>0$, such that $  |\langle\eta_k| e^{-i H t} |\eta_k\rangle|<1$.\\

Note that the last part of theorem means that the greatest common divisor of integers $\{m_k=\tau/\tau_k\}$ is one.  As we mentioned before,  the first part of the theorem was conjectured by Toth and Petz \cite{toth2013extremal}, and is proven by Yu \cite{yu2013quantum}.  In the paper we showed that this part follows from theorem 3. Here, we prove the second part of the theorem, which puts a constraint on the period of states in the optimal ensemble.

\begin{proof}
Let $H= \sum_E E \Pi_E$ be the spectral decomposition of Hamiltonian $H$.  The fact that $e^{-i H \tau}\rho e^{i H \tau}=\rho$ implies that for any two energy levels $E_1$ and $E_2$ if $\Pi_{E_1} \rho \Pi_{E_2}\neq 0$, then $E_1-E_2=2\pi m/\tau $ for an integer $m$.  Based on the criterion that $\Pi_{E_1} \rho \Pi_{E_2}$ is zero or not we can divide the energy levels into disjoint  partitions, such that (i) the energy levels in each partition are separated with energy gaps $2\pi m/\tau $ for an integer $m$, and (ii) for any two energy levels $E_1$ and $E_2$ in two different partitions 
 $\Pi_{E_1} \rho \Pi_{E_2}= 0$. 
 
 Suppose we label disjoint partitions with  $r$ and let $P_{r}$ be the projector to the subspace spanned by the energy level in the partition  $r$, i.e., each $P_r$ is the  sum of $\Pi_E$ for all $E$ belonging to the same partition $r$.  Note that there is no coherence between different partitions. That is 
\be\label{cph179}
\sum_r P_{r} \rho P_{r} = \rho\ .
 \ee
 Let  $\{q_k,|\eta_k\rangle\}$ be an optimal ensemble  satisfying 
 \be\label{wetyui}
F_H(\rho)= 4\times \sum_k q_k V_H(\eta_k)\ .
\ee 
Now  we define a new ensemble of pure states, which is obtained from this ensemble by removing coherence between different partitions defined above. Namely  the ensemble 
 \be
 \left\{\ \tilde{q}_{k,r}=q_k \langle\eta_k|P_{r} |\eta_k\rangle\ ,\ |\tilde{\eta}_{k,r}\rangle= \frac{P_r |\eta_k\rangle}{\sqrt{\langle\eta_k|P_{r} |\eta_k\rangle}} \ \right\}_{k,r}
 \ee
 where state $\frac{P_r |\eta_k\rangle}{\sqrt{\langle\eta_k|P_{r} |\eta_k\rangle}}$ happens with probability $ \tilde{q}_{k,r}$. This ensemble can be thought of  as the ensemble obtained from the optimal ensemble $\{q_k,|\eta_k\rangle\}$ by measuring in the basis $\{P_{r} \}_r$. It can be easily seen that 
\begin{enumerate}
\item Eq.(\ref{cph179}) together with $\sum_k q_k |\eta_k\rangle\langle\eta_k|=\rho$ imply 
\be
\sum_{k,r} \tilde{q}_{k,r}  |\tilde{\eta}_{k,r}\rangle\langle\tilde{\eta}_{k,r}|=\rho\ .
\ee

\item Concavity of variance implies
\begin{align}
\sum_{r} \tilde{q}_{k,r}  V_H(|\tilde{\eta}_{k,r}\rangle)&=q_k \sum_r  \langle\eta_k|P_{r} |\eta_k\rangle  V_H(|\tilde{\eta}_{k,r}\rangle) \le q_k   V_H(|{\eta}_{k}\rangle)\ .
\end{align}
It follows that the average variance for the ensemble $
 \left\{\ \tilde{q}_{k,r} ,\ |\tilde{\eta}_{k,r}\rangle\right\}_{k,r}$ satisfies 
\begin{align}
\sum_{k,r} \tilde{q}_{k,r}  V_H(|\tilde{\eta}_{k,r}\rangle)&=\sum_k q_k \sum_r  \langle\eta_k|P_{r} |\eta_k\rangle  V_H(|\tilde{\eta}_{k,r}\rangle) \\ &\le \sum_k q_k   V_H(|{\eta}_{k}\rangle)\\ &=\frac{F_H(\rho)}{4}\\ &\le \sum_{k,r} \tilde{q}_{k,r}  V_H(|\tilde{\eta}_{k,r}\rangle)\  ,
\end{align}
where to get the second line we have used the fact that variance is a concave function, the third line follows from Eq.(\ref{wetyui}), and the last line follows from convexity of QFI. We conclude that
\be
\sum_{k,r} \tilde{q}_{k,r}  V_H(|\tilde{\eta}_{k,r}\rangle)= \frac{F_H(\rho)}{4}\  .
\ee
\item Since for each projector $P_r$ the difference between any two energy levels is an integer multiple of $2\pi/\tau$, for any $k$ and $r$ the period of state $ |\tilde{\eta}_{k,r}\rangle=\frac{P_r |\eta_k\rangle}{\sqrt{\langle\eta_k|P_{r} |\eta_k\rangle}}$ is $\tau_k=\tau/m_k$ for an integer $m_k$, or is zero, i.e., $ |\tilde{\eta}_{k,r}\rangle$ is an energy eigenstate.
\item Since for any time $t<\tau$, $e^{-i H t}\rho e^{i H t}\neq \rho$, we conclude that for any $t<\tau$ there should be at least one pure state $|\tilde{\eta}_{k,r}\rangle$ such that  $|\langle\tilde{\eta}_{k,r}|e^{-i H t}|\tilde{\eta}_{k,r}\rangle|\neq 1$. Equivalently, this means that the greatest common divisors of integers $m_k=\tau/\tau_k$ is one. 

 \end{enumerate}
This completes the proof.
 \end{proof}

\color{black}

\newpage

\section{Monotonicity of Fisher information in  approximate asymptotic transformations   }\label{Sec:Fisher-mon}

Consider a pair of systems labled as the input and output systems,  with the Hilbert spaces  $\mathcal{H}_\text{in}$ and  $\mathcal{H}_\text{out}$ and the corresponding Hamiltonians  $H_\text{in}$ and $H_\text{out}$. Define the superoperators $\mathcal{U}_\text{in}(t)$  and $\mathcal{U}_\text{out}(t)$ to be the time translations generated by  $H_\text{in}$ and $H_\text{out}$, i.e.
\be
\mathcal{U}_\text{in}(t)[\cdot]=e^{-i H_\text{in}t}(\cdot) e^{i H_\text{in}t}\ \ \ ,\ \ \ \ \ \ \ \ \mathcal{U}_\text{out}(t)[\cdot]=e^{-i H_\text{out} t}(\cdot)e^{i H_\text{out} t}\ .
\ee

\begin{theorem}\label{Thm:mono}
Suppose there exists a sufficiently large integer $n_0$ such that for all  integer $n\ge n_0$, there exists a  CPTP map $\mathcal{E}_n$  that transforms $n$ copies of the input system to $m=\lceil R n\rceil$
 copies of the output system, such that  (i) $\mathcal{E}_n$ satisfies the covariance condition
\be\label{cov153}
\forall t:\ \  \mathcal{U}^{\otimes \lceil R n\rceil}_\text{out}(t)\circ \mathcal{E}_n=  \mathcal{E}_n\circ \mathcal{U}^{\otimes n}_\text{in}(t)\ ,
\ee 
and (ii) maps the input state $\rho^{\otimes n}$ to $\sigma^{\otimes \lceil R n\rceil}$  with error bounded by $\delta$, such that
\be
\frac{1}{2} \big\|\mathcal{E}_n(\rho^{\otimes n})-\sigma^{\otimes \lceil R n\rceil}\big\|_1 \le \delta\ .
\ee
Then, 
\be
T^{{T}/{(1-T)}}- T^{{1}/{(1-T)}}  \le 4 \sqrt{\delta}\ , 
\ee
where 
\be
T=\frac{F_{H_\text{in}}(\rho)}{R\times  F_{H_\text{out}}(\sigma)}\ . 
\ee 
 \end{theorem}
It can be easily seen that function $g(x)=x^{{x}/{(1-x)}}- x^{{1}/{(1-x)}}$ is positive in the interval $x\in[0,1)$. This means that for  $T={F_H(\rho)}/({R\times  F_H(\sigma)})<1$, $ \sqrt{\delta}$ is lower bounded by a positive number. That is if $F_H(\rho)< R\times  F_H(\sigma)$  then the error $\delta$ cannot be arbitrary small. In summary, we conclude that if  there exists a sequence of TI operations  that convert copies of the input systems to the copies of the output systems  with rate $R(\rho\rightarrow\sigma)$,  with a vanishing error in the trace distance, then 
  \be
R(\rho\rightarrow\sigma)\le \frac{F_{H_\text{in}}(\rho)}{F_{H_\text{out}}(\sigma)}\ .
  \ee

Before presenting the proof, we recall the Fuchs-van de Graaf inequality \cite{fuchs1999cryptographic, nielsen2000quantum, wilde2013quantum}: For any pair of density operators $\rho_1$ and $\rho_2$ it holds that
\be\label{Fuchs}
1-\sqrt{\text{Fid}(\rho_1,\rho_2)} \le \frac{1}{2}\|\rho_1-\rho_2\|_1\le  \sqrt{1-\text{Fid}(\rho_1,\rho_2)}\ ,
\ee
where $\|\cdot \|_1$ denotes the l-1 norm, that is the sum of the singular values.  Using the properties of fidelity and Bures metric, in Sec. \ref{Sec:proof21} we  prove the following lemma, which will be used in the proof of theorem \ref{Thm:mono}.
\begin{lemma} \label{lem:fid}
For any pairs of states $\tau_1$ and $\tau_2$ and unitary $U$ it holds that
\bes\label{Bures}
\begin{align}
\Big|\sqrt{\text{Fid}(U \tau_1 U^\dag , \tau_1)}- \sqrt{\text{Fid}(U \tau_2 U^\dag , \tau_2)}\Big| &\le 4\sqrt{1-\sqrt{\text{Fid}(\tau_1,\tau_2)}}\le 4 \sqrt{\frac{1}{2}\|\tau_1-\tau_2\|_1}\ .
\end{align}
\ees
\end{lemma}

\subsection{Proof of theorem \ref{Thm:mono}}

To simplify the notation we assume the input and output system Hamiltonians are identical and they are both denoted by $H$. Generalizing the result to the case where these systems are different is straightforward.

For any $n$ let $m=\lceil R n\rceil$ be the number of copies of the output systems.  Let $\sigma_m=\mathcal{E}_n(\rho^{\otimes n})$ be the actual output state and  $\sigma_m(\Delta t)$ be the time-evolved version of  $\sigma_m$, i.e.
\begin{align}
\sigma_m(\Delta t)&=\mathcal{U}^{\otimes m}(\Delta t)[\sigma_m]=(e^{-i H \Delta t})^{\otimes m} \sigma_m (e^{i H \Delta t})^{\otimes m}  \ .
\end{align}
Here, $\Delta t$ is a parameter whose value will be fixed later. Similarly, let $\sigma^{\otimes m}(\Delta t)=(e^{- i H \Delta t}  \sigma e^{i H \Delta t})^{\otimes m}$ be the time-evolved version of state $\sigma^{\otimes m}$. Since the  operation  $\mathcal{E}_{n}$ is TI, i.e. satisfies the covariance condition in Eq.(\ref{cov153}),  for any $\Delta t$ it maps state $\rho^{\otimes n}(\Delta t)$ to state $\sigma_m(\Delta t)$. To summarize
\bes
\begin{align}
\mathcal{E}_{n}(\rho^{\otimes n})&=\sigma_m\\ \mathcal{E}_{n}(\rho(\Delta t)^{\otimes n})&=\sigma_m(\Delta t)\ .
\end{align}
\ees
Then, 
\be\label{Eq2}
\big[\text{Fid}(\rho, {\rho(\Delta t)})\big]^{n/2}=  \sqrt{\text{Fid}(\rho^{\otimes n}, {\rho(\Delta t)}^{\otimes n} )}  \le \sqrt{\text{Fid}(\sigma_m, \sigma_m(\Delta t))} \ ,
\ee
where the equality follows from the multiplicativity of fidelity under tensor products and the bound follows from the  monotonicity of Fidelity under CPTP maps.

Applying lemma \ref{lem:fid} to states $\sigma^{\otimes m}$ and $\sigma_m$, we find
\begin{align}
 \sqrt{\text{Fid}\left(\sigma_m(\Delta t),\sigma_m\right)}-\sqrt{\text{Fid}\left(\sigma(\Delta t)^{\otimes m}, \sigma^{\otimes m}\right)} &\le 4\sqrt{1-\sqrt{\text{Fid}(\sigma_m,\sigma^{\otimes m})}}\ .
\end{align}
Using the multiplicativity of fidelity for tensor products and applying Eq.(\ref{Eq2}),  this implies
\begin{align}
\big[\text{Fid}(\rho, {\rho(\Delta t)})\big]^{n/2}\le \sqrt{\text{Fid}(\sigma_m(\Delta t),\sigma_m)} &\le 4\sqrt{1-\sqrt{\text{Fid}(\sigma_m,\sigma^{\otimes m})}}+ \big[\text{Fid}(\sigma(\Delta t), \sigma)\big]^{m/2}\ .
\end{align}
Then, choosing $m=\lceil R n\rceil$ we find  
\begin{align}
\big[\text{Fid}({\rho(\Delta t)}, \rho)\big]^{n/2} - \big[\text{Fid}(\sigma(\Delta t), \sigma)\big]^{\lceil R n\rceil/2} &\le 4\sqrt{1-\sqrt{\text{Fid}(\sigma_m,\sigma^{\otimes m})}}\\ &\le 4\sqrt{\frac{1}{2}\|\sigma_m-\sigma^{\otimes m}\|_1} \ ,
\end{align}
where to  get the second line we
have used Fuchs-van de Graaf inequality in Eq.(\ref{Fuchs}). By the assumption of the theorem,  for $n\ge n_0$, the trace distance $\frac{1}{2}\|\sigma_m-\sigma^{\otimes m}\|_1$ is bounded by $\delta$. This implies
\begin{align}\label{kdkd45}
\big[\text{Fid}({\rho(\Delta t)}, \rho)\big]^{n/2} - \big[\text{Fid}(\sigma(\Delta t), \sigma)\big]^{\lceil R n\rceil/2} &\le 4\sqrt{\delta} \ .
\end{align}
Next, we take $\Delta t=t/\sqrt{n}$ for arbitrary fixed $t\in\mathbb{R}$ and consider the limit of large $n$, where $\Delta t$ goes to zero. Consider the Taylor expansion of $\text{Fid}( {\rho(\Delta t)}, \rho)$, as a function of $\Delta t$ around $\Delta t=0$.  Since $\text{Fid}( {\rho(\Delta t)}, \rho)$ is an even function of $\Delta t$,  its odd derivatives with respect to $\Delta t$ vanishes. Furthermore, the second derivative of function 
$\text{Fid}( e^{-i H \Delta t}\rho e^{iH \Delta t}, \rho)$ 
 with respect to $\Delta t$ is $1/4$ times the QFI for the family of states $\{e^{-i H \Delta t} \rho e^{i H \Delta t} \}$ and parameter $\Delta t$ (Theorem 6.3  \cite{hayashi2006quantum}). In other words,
for infinitesimal  $\Delta t$, 
\be
\text{Fid}({\rho(\Delta t)}, \rho)=\text{Fid}( e^{-i H \Delta t}\rho e^{iH \Delta t}, \rho)=1-\frac{\Delta t^2}{4} {F}_H(\rho)+\mathcal{O}(\Delta t^4)\ ,
\ee
where $\mathcal{O}(\Delta t^4)$ denotes terms of order $\Delta t^4$ and higher.  Therefore, for $\Delta t=t/\sqrt{n}$,  in the limit of large $n$ we  find
\begin{align}
\text{Fid}({\rho(\frac{t}{\sqrt{n}})}, \rho)=1-\frac{t^2}{4 n} {F}_H(\rho)+\mathcal{O}( \frac{t^4}{n^2}) \ .
\end{align}
Then, using the fact that $\lim_{n\rightarrow \infty} (1-x/n)^n=e^{-x}$, we find that in the limit $n$ goes to infinity, $\text{Fid}^{n/2}(\rho, {\rho(\Delta t)})$ in  the left-hand side of Eq.(\ref{kdkd45})  converges to 
\be\label{Eq6401}
\lim_{n\rightarrow \infty} \Big[\text{Fid}(\rho, {\rho(\frac{t}{\sqrt{n}})})\Big]^{n/2}=e^{-t^2  F_H(\rho)/8}\ .
\ee
Similarly, $\big[\text{Fid}(\sigma(\Delta t), \sigma)\big]^{\lceil R n\rceil/2} $ converges to $e^{-\frac{1}{8}  R t^2  F_H(\sigma)} $. Therefore, Eq.(\ref{kdkd45}) implies
\begin{align}\label{kdkd46}
e^{-\frac{1}{8}t^2  F_H(\rho)} - e^{-\frac{1}{8}  R t^2  F_H(\sigma)} \le 4\sqrt{\delta}\ .
\end{align}
This bound holds for arbitrary $t\in\mathbb{R}$. Te strongest bound is achieved when the left-hand side is maximized, which happens for  
\be
t^2=\frac{-8}{R  \times F_H(\sigma)-F_H(\rho)}\times \log \frac{F_H(\rho)}{R \times   F_H(\sigma)}\ .
\ee
In this case, Eq.(\ref{kdkd46}) implies 
\begin{align}
g(T) \le 4\sqrt{\delta}\ ,
\end{align}
where $T=\frac{F_H(\rho)}{R\times  F_H(\sigma)}$ and $g(x)\equiv x^{{x}/{(1-x)}}- x^{{1}/{(1-x)}}$. 

To complete the proof of theorem, in the following we prove lemma \ref{lem:fid}. 

\subsection{Proof of lemma \ref{lem:fid} }\label{Sec:proof21}


We first  recall some useful properties of Fidelity and the Bures distance. Recall that fidelity of two states $\rho_1$ and $\rho_2$ is defined as $\text{Fid}(\rho_1,\rho_2)=\|\sqrt{\rho_1}\sqrt{\rho_2}\|^2_1=\Tr(\sqrt{\sqrt{\rho_1}\rho_2 \sqrt{\rho_1}})^2$.  Fidelity is not  a distance but it is closely related to the Bures distance, via the relation 
\be
B(\rho_1,\rho_2)=\sqrt{2[1-\sqrt{\text{Fid}(\rho_1,\rho_2)}]}\ .
\ee
This function satisfies all the properties of a distance. In particular, it is symmetric, i.e.,  $B(\rho_1,\rho_2)=B(\rho_2,\rho_1)$, and satisfies the triangle inequality, i.e.
\be
B(\rho_1,\rho_2)+B(\rho_2,\rho_3)\ge B(\rho_1,\rho_3)\ .
\ee
Furthermore, it is  invariant under unitary transformations, i.e. $B(\rho_1,\rho_2)=B(U\rho_1U^\dag,U\rho_2U^\dag)$, 
which can be easily seen using its relation with fidelity. 

Now, we are ready to present the proof of the lemma.  Using the triangle   inequality twice we find 
\begin{align}
 B(U \tau_2 U^\dag ,U \tau_1 U^\dag )+B(U \tau_1 U^\dag , \tau_1)+ B(\tau_1,\tau_2) \ge  B(U \tau_2 U^\dag , \tau_2) .
 \end{align}
 Let  $\eta\equiv B(U \tau_1 U^\dag ,U \tau_2 U^\dag )=B(\tau_1,\tau_2)$. Then, the above inequality can be rewritten as  
$B(U \tau_1 U^\dag , \tau_1)\ge B(U \tau_2 U^\dag ,\tau_2)-2 \eta\ $, which in turn implies
 \begin{align}
 B^2(U \tau_1 U^\dag , \tau_1)&\ge B^2(U \tau_2 U^\dag , \tau_2)-4 \eta B(U \tau_2 U^\dag , \tau_2)\\ &\ge B^2(U \tau_2 U^\dag , \tau_2)-4\sqrt{2} \eta\ \ ,
 \end{align}
 where we have used the fact that Bures metric  is bounded by $\sqrt{2}$. This, in turn implies 
 \be
 1-\sqrt{\text{Fid}(U \tau_1 U^\dag , \tau_1)}\ge 1-\sqrt{\text{Fid}(U \tau_2 U^\dag , \tau_2)}-2 \sqrt{2} \eta\  ,
 \ee
  and so 
\be
2 \sqrt{2} \eta \ge \sqrt{\text{Fid}(U \tau_1 U^\dag , \tau_1)}- \sqrt{\text{Fid}(U \tau_2 U^\dag , \tau_2)}\  .
 \ee
  Exchanging $\tau_1$ and $\tau_2$ we also find  $ 2 \sqrt{2} \eta \ge \sqrt{\text{Fid}(U \tau_2 U^\dag , \tau_2)} -\sqrt{\text{Fid}(U \tau_1 U^\dag , \tau_1)}  $. We conclude that  
  $|\sqrt{\text{Fid}(U \tau_1 U^\dag , \tau_1)}- \sqrt{\text{Fid}(U \tau_2 U^\dag , \tau_2)}| \le 4\sqrt{1-\sqrt{\text{Fid}(\tau_1,\tau_2)}}$.  
   Finally, combining this with Fuchs-van de Graaf  in Eq.(\ref{Fuchs}) proves the lemma. 

\newpage
 
 \section{Quantum Fisher Information as the coherence cost: iid regime (Proof of theorem \ref{Thm0})}\label{Sec:QFI:iid}

In this section we prove that in the iid regime the coherence cost of preparing any state is determined by its Quantum Fisher information.\\

\noindent{\textbf{Restatement of Theorem} \ref{Thm0}: } 
Consider a system with Hamiltonian $H$ and state $\rho$ with period $\tau$, with a finite-dimensional Hilbert space. Consider a two-level system with state $|\Theta\rangle_{\text{c-bit}}=(|0\rangle+|1\rangle)/\sqrt{2}$ and Hamiltonian $H_\text{c-bit}=\pi \sigma_z/\tau$. Then, for any $R>F_H(\rho)/F_{\text{c-bit}}= (\tau/2\pi)^2 F_H(\rho) $, and integer $n$, there exists a TI operation that converts $\Theta_\text{c-bit}^{\otimes\lceil R n\rceil}$ to a state whose trace distance from $\rho^{\otimes n}$ is bounded by $\epsilon_n>0$, and $\epsilon_n\rightarrow 0$ in the limit $n$ goes to infinity, i.e.
\be\nonumber
\Theta_\text{c-bit}^{\otimes\lceil R n\rceil} \xrightarrow{TI}\stackrel{\epsilon_n}{ \approx}  \rho^{\otimes n}\ \ \ \  \text{as } \ \ \ \ n\rightarrow\infty, \ \ \epsilon_n\rightarrow 0\ . 
\ee
Furthermore, for any $R<F_H(\rho)/F_{\text{c-bit}}= (\tau/2\pi)^2 F_H(\rho) $ the above transformation is not possible with a vanishing error $\epsilon_n$.\\

 Another way to phrase this result is in terms of the coherence cost of system with state $\rho$ and Hamiltonian $H$:  the coherence cost of state $\rho$ is given by
\be
C^\text{TI}_c(\rho) = \frac{F_H(\rho)}{F_{\text{c-bit}}}=(\frac{\tau}{2\pi})^2 \times F_H(\rho)  \ .
\ee

The proof of the second part, i.e. $C^\text{TI}_c(\rho) \ge F_H(\rho)/F_{\text{c-bit}}$ follows from our result in Sec.\ref{Sec:Fisher-mon}, and in particular, theorem \ref{Thm:mono}.  In this section we prove the first part of  theorem \ref{Thm0}, i.e., we show that  $C^\text{TI}_c(\rho) \le F_H(\rho)/F_{\text{c-bit}}$. The proof uses theorem \ref{Thm:Fisher}. According to this theorem 
QFI is four times the \emph{convex roof} of the variance, i.e.
\be\label{Eq1252}
F_H(\rho)= \min_{\{q_k,\eta_k\}}\sum_k q_k F_H(\eta_k)= 4\times \min_{\{q_k,\eta_k\}}\sum_k q_k V_H(\eta_k)\ ,
\ee
where the minimization is over the set of all ensembles of pure states $\{q_k,|\eta_k\rangle\}$ satisfying $\sum_k q_k |\eta_k\rangle\langle\eta_k|=\rho$. 

To prove theorem \ref{Thm0}, we also use the following lemma, which can be shown using the results of Sec. \ref{Sec:app1} on pure state conversions.

\begin{lemma}\label{app:lem12}
For a system with Hamiltonian $H$,
consider a finite set of pure states $\mathbb{S}=\{|\psi_k\rangle\}_k$ with the property that  $|\langle\psi_k|e^{-iH \tau}|\psi_k\rangle|=1$.    Furthermore, suppose for any $0<t<\tau$ there is, at least, one state $|\psi_k\rangle\in \mathbb{S}$  such that  $|\langle\psi_k|e^{-iH t}|\psi_k\rangle|<1$. Let $\{r_k >0\}$ be an arbitrary set of positive real numbers. 

Then, for any integer $m$ there exists a TI operation that converts   $\lceil m R \rceil$  copies of system with state  $|\Theta\rangle_{\text{c-bit}}=(|0\rangle+|1\rangle)/\sqrt{2}$  and Hamiltonian $H_\text{c-bit}=\pi\sigma_z/\tau$ to state $|\Psi(m)\rangle=\bigotimes_{k\in \mathbb{S}} |\psi_k\rangle^{\otimes \left \lceil{r_k m}\right \rceil}$ with an error $\epsilon_m$ that vanishes in the limit $m$ goes to infinity,  provided that  $R> \sum_k r_k V_H(\psi_k)/V_{c-bit}=\sum_k r_k V_H(\psi_k)\tau^2/\pi^2$, where $V_{c-bit}=(\pi/\tau)^2$. In summary
\be\label{trans591}
R> \sum_{k\in S} r_k \frac{V_H(\psi_k)}{V_{c-bit}}\ \  \Longrightarrow\ \   \Theta_\text{c-bit}^{\otimes\lceil R m\rceil} \xrightarrow{TI}\stackrel{\epsilon_m}{ \approx} |\Psi(m)\rangle=\bigotimes_{k\in \mathbb{S}} |\psi_k\rangle^{\otimes \left \lceil{r_k m}\right \rceil}\  \ \  \text{as } m\rightarrow\infty,\  \epsilon_m\rightarrow 0\ . 
\ee
\end{lemma}
\begin{proof}
First, we temporarily assume all $r_k>0$ are rational numbers and show that the result follows  from theorem \ref{Thm:main:app}: 
Since $\mathbb{S}$ is a finite set, there exists a finite integer $M$ such that $M r_k$ is an integer for all $k$. Now consider state 
\be
|{\Psi}(M)\rangle=\bigotimes_{k\in S} |\psi_k\rangle^{\otimes (M r_k)} .
\ee
It can be easily seen that this state  has energy variance
 \be
M \sum_k r_k V_H(\psi_k)\ ,
\ee  
and period $\tau$. To see the latter note that all states $\psi_k$ satisfies $|\langle\psi_k|e^{-iH \tau}|\psi_k\rangle|=1$ and for any $t<\tau$ there is, at least, one state $|\psi_k\rangle$  such that  $|\langle\psi_k|e^{-iH t}|\psi_k\rangle|<1$, we conclude that 
\be
 \tau=\inf_t \Big\{t>0:\ \   \left|\langle{\Psi}(M)| (e^{-i H t })^{\otimes (M \sum_k r_k)} |{\Psi}(M)\rangle\right|=1 \Big\}\ .
 \ee
 Therefore, we can apply the results of section \ref{Sec:app1} and in particular, theorem \ref{Thm:main:app} which implies that by consuming c-bits at rate 
\be
R'> M \sum_k r_k \frac{V_H(\psi_k)}{V_\text{c-bit}} \ ,
\ee 
per copy it is possible to prepare copies of   $|\Psi(M)\rangle$ using TI operations.  More precisely, there exists a sequence of TI operations  that implements the state conversion   
\be
|\Theta\rangle_{\text{c-bit}}^{\otimes\lceil R' n\rceil}=|\Theta\rangle_{\text{c-bit}}^{\otimes\lceil \frac{R'}{M} (M\times n)\rceil}=|\Theta\rangle_{\text{c-bit}}^{\otimes\lceil \frac{R'}{M} \times m\rceil}   \xrightarrow{TI}\stackrel{\epsilon_n}{ \approx} |\Psi(M)\rangle^{\otimes n}=|\Psi(n\times M)\rangle=|\Psi(m)\rangle   \ \  \text{as } n\rightarrow\infty,\  \epsilon_n\rightarrow 0\ ,
\ee
where $m=M\times n$.  This proves Eq.(\ref{trans591})  for the subsequence of integers that are multiple of $M$.  It can be easily seen that the result also holds for general integer $m$: Clearly,  by discarding subsystems, which is a TI operation, we can convert $|\Psi(m)\rangle$  to $|\Psi(m')\rangle$ for any $m'\le m$. Therefore, to generate $|\Psi(m)\rangle$ for $m$ which is not a multiple of $M$, we can  generate $|\Psi(M)\rangle^{\otimes n}$  for $n=\lceil m/M\rceil$, and then discard the additional subsystems. This way we can implement the state conversion
\be
\Theta_\text{c-bit}^{\otimes\lceil R m\rceil} \xrightarrow{TI}\stackrel{\epsilon_m}{ \approx} |\Psi(m)\rangle=\bigotimes_{k\in S} |\psi_k\rangle^{\otimes \left \lceil{r_k m}\right \rceil}\  \ \  \text{as } m\rightarrow\infty,\  \epsilon_m\rightarrow 0\ , 
\ee
provided that
\be
\lceil R m\rceil > R' \times \lceil \frac{m}{M}\rceil> (M \lceil \frac{m}{M}\rceil) \times  \sum_k r_k \frac{V_H(\psi_k)}{V_\text{c-bit}} 
\ee
In the limit $m\rightarrow \infty$,  this is equivalent to  
\be
R >  \sum_k r_k \frac{V_H(\psi_k)}{V_\text{c-bit}}    \ .
\ee

This proves the result for the special case where all coefficients  $r_k>0$ are rational numbers. To extend the result to the case of irrational numbers, for each $r_k$ we choose a rational number $\tilde{r}_k\ge r_k$. Then, applying the above argument we find that the state conversion in Eq.(\ref{trans591}) is possible with any rate $R > \sum_k \tilde{r}_k V_H(\psi_k) $. Since the rational number $\tilde{r}_k$ can be arbitrary close to $r_k$, we conclude that for any $R> \sum_k r_k V_H(\psi_k) $  the state conversion in  Eq.(\ref{trans591}) can be implemented by a TI operation with an error which vanishes in the limit $m$ goes to infinity. This completes the proof.

 
\end{proof}

\subsection{Proof of theorem \ref{Thm0}}

 Finally, we present the proof of  theorem \ref{Thm0}.  From theorem \ref{Thm:Fisher}  we know that there exists a finite ensemble $\{(q_k, |\psi_k\rangle): k\in \mathbb{S}\}$ with density operator  $\rho=\sum_{k\in \mathbb{S}} q_k |\psi_k\rangle\langle\psi_k|$,  satisfying 
  \begin{enumerate}
 \item  $F_H(\rho)= \sum_k q_k F_H(\psi_k)= 4\times \sum_k q_k V_H(\psi_k)$\ .
 \item For all $k$, $|\langle\psi_k| e^{-i H \tau} |\psi_k\rangle|=1$, and for any $0<t<\tau$, there is at least one state $ |\psi_k\rangle$ in this ensemble such that $|\langle\psi_k| e^{-i H t} |\psi_k\rangle|< 1$.
 \end{enumerate}
To see the latter statement, note that if 
there is a $t_0\in(0,\tau)$ such that $|\langle\psi_k| e^{-i H t_0} |\psi_k\rangle|=1$ for all $k\in \mathbb{S}$, then
\be
e^{-i H t_0} \rho e^{i H t_0}=\sum_k q_k\  e^{-i H t_0}|\eta_k\rangle\langle\eta_k| e^{i H t_0}=\sum_k q_k\ |\eta_k\rangle\langle\eta_k|=\rho\ ,
\ee
which contradicts with the assumption that the period of $\rho$ under $H$ is $\tau$.

Next, consider $m$ copies of $\rho$,  i.e.,  state 
\be
\rho^{\otimes m}=\Big(\sum_{k\in S} q_k |\psi_k\rangle\langle\psi_k|\Big)^{\otimes m} =\sum_{\textbf{k} } q_{\textbf{k} } |\psi_\textbf{k}\rangle\langle\psi_\textbf{k}|\ ,
 \ee
where $\textbf{k}=k_1\cdots k_m\in \mathbb{S}^m$,    $q_{\textbf{k} }=q_{k_1}\cdots q_{k_m}$  and $|\psi_\textbf{k}\rangle=|\psi_{k_1}\rangle\otimes \cdots \otimes |\psi_{k_m}\rangle$.  

For any  $l\in\mathbb{S}$, let $n_l(\textbf{k})$ be the number of occurrence of $l$ in the string $\textbf{k}=k_1\cdots k_m\in \mathbb{S}^m$.  Then, for any $\delta >0$, we define the set of $\delta$-typical strings as
 \be
 \mathcal{T}_\delta\equiv \{\textbf{k}=k_1\cdots  k_m\in \mathbb{S}^m\big|\  \forall l\in \mathbb{S}:\ |\frac{n_l(\textbf{k})}{m}-q_l|\le \delta \}\ .
 \ee 
In other words, $ \mathcal{T}_\delta$ is the set of all strings for which the relative frequency of any $l\in \mathbb{S}$ is between $q_l-\delta$ and $q_l+\delta$. 

Now consider the decomposition of state $\rho^{\otimes m}$ as
\be
\rho^{\otimes m}=\sum_{\textbf{k} }  q_{\textbf{k} }\ |\psi_\textbf{k}\rangle\langle\psi_\textbf{k}|\ =\sum_{\textbf{k}\in\mathcal{T}_\delta }  q_{\textbf{k} }\ |\psi_\textbf{k}\rangle\langle\psi_\textbf{k}|+\sum_{\textbf{k}\notin\mathcal{T}_\delta }  q_{\textbf{k} }\ |\psi_\textbf{k}\rangle\langle\psi_\textbf{k}|\  .
 \ee
Based on this decomposition we can 
define a TI operation for preparing a state close to $\rho^{\otimes m}$: First, we sample $\textbf{k}\in \mathbb{S}^n$ with probability $q_{\textbf{k} }$. If $\textbf{k}$ is in the typical set $\mathcal{T}_\delta$, then we prepare state $|\psi_\textbf{k}\rangle$. Otherwise, we prepare  a fixed incoherent  state $\sigma_\text{inv}$, e.g., the maximally  mixed state. 
The resulting state is
\be
\tilde{\rho}_m=\sum_{\textbf{k}\in\mathcal{T}_\delta }   q_{\textbf{k} } |\psi_\textbf{k}\rangle\langle\psi_\textbf{k}|+ p_\text{err}\  \sigma_\text{inv}\ ,
 \ee
 where  
\be
p_\text{err}=\sum_{\textbf{k}\notin\mathcal{T}_\delta }  p_{\textbf{k} }=1-\sum_{\textbf{k}\in\mathcal{T}_\delta }  p_{\textbf{k} }\ .
\ee
  Then, using the standard typicality arguments, we know that for any fixed $\delta>0$, in the limit $m$ goes to infinity, the probability $p_\text{err}$ goes to 0, which means with probability approaching 1 the sampled string is in the typical set $\mathcal{T}_\delta$. This, in turn implies that  the trace distance between the output   state $\tilde{\rho}_m$  and the desired state $\rho^{\otimes m}$ vanishes, i.e.
 \be\label{app:limit}
\lim_{m\rightarrow \infty}\frac{1}{2}\|\rho^{\otimes m}-\tilde{\rho}_m\|_1 =0\ .
 \ee
Therefore, in the following we focus on the coherence cost of preparing $|\psi_\textbf{k}\rangle$ for $\textbf{k}$ in the typical set $\mathcal{T}_\delta$.  Up to a permutation, which is a TI unitary, this state  can be written as
 \be
\bigotimes_{l\in S} |\psi_l\rangle^{\otimes n_l(\textbf{k})}\ ,
\ee
where, the typicality of string $\textbf{k}$ implies
\be
n_l(\textbf{k}) \le m \times (q_l+\delta) \ .
\ee
Therefore, to obtain $|\psi_\textbf{k}\rangle$ it suffices to prepare 
\be
|\Psi(m)\rangle=\bigotimes_{l\in \mathbb{S}} |\psi_l\rangle^{\otimes \left \lceil{m\times (q_l+\delta) }\right \rceil}\ ,
\ee
and then, possibly discard some
 subsystems and permute the remaining ones. 
 
Then, we apply lemma  \ref{app:lem12}. Note that for all $l\in \mathbb{S}$, we have $|\langle\psi_l|e^{-iH \tau}|\psi_l\rangle|=1$, and for any $0<t<\tau$ there is, at least, one state $|\psi_l\rangle$  such that  $|\langle\psi_l|e^{-iH t}|\psi_l\rangle|<1$. Therefore, all the assumptions of this lemma are satisfied. Applying this lemma we conclude that if $\textbf{k}$ is in the typical set, then  there exists a sequence of TI operations that implements the state conversion 
\be
  \Theta_\text{c-bit}^{\otimes\lceil R m\rceil} \xrightarrow{TI}\stackrel{\epsilon}{ \approx}|\psi_\textbf{k}\rangle\  \ \  \text{as } m\rightarrow\infty, \epsilon\rightarrow 0\ ,
\ee
provided that
 \be
R > \sum_l (q_l+\delta)  \frac{V_H(\psi_l)}{V_{c-bit}}\ .
\ee
Note that since  $ \sum_l q_l V_H(\psi_l)=F_H(\rho)$ and $V_H(\psi)\le \|H\|^2$,  the latter condition is satisfied if
\be
R > \frac{F_H(\rho)+\delta \|H\|^2}{F_{c-bit}}=(\frac{\tau}{2\pi})^2  \times [F_H(\rho)+\delta \|H\|^2]\ ,
\ee
where $F_{c-bit}=4 V_{c-bit}= (2\pi/\tau)^2$.

This, in turns implies  there exists a sequence of  TI operations that implements 
\be
 \Theta_\text{c-bit}^{\otimes\lceil R m\rceil} \xrightarrow{TI}\stackrel{\epsilon}{ \approx}\tilde{\rho}_m=\sum_{\textbf{k} }  q_{\textbf{k} } |\psi_\textbf{k}\rangle\langle\psi_\textbf{k}|+ p_\text{err}\  \sigma_\text{inv}   \ \  \text{as } m\rightarrow\infty, \epsilon\rightarrow 0\ . 
\ee
Combining this with Eq.(\ref{app:limit}),  and using the fact that $\delta>0$ can be chosen arbitrarily small, we find that for any $R>(\frac{\tau}{2\pi})^2 F_H(\rho)$ there exists a sequence of TI operations that implements
\be
 \Theta_\text{c-bit}^{\otimes\lceil R m\rceil} \xrightarrow{TI}\stackrel{\epsilon}{ \approx}\rho^{\otimes m}   \ \  \text{as } m\rightarrow\infty, \epsilon\rightarrow 0\ .
\ee
This proves $C^\text{TI}_c(\rho) \le (\frac{\tau}{2\pi})^2 F_H(\rho)$. The proof of $C^\text{TI}_c(\rho) \ge (\frac{\tau}{2\pi})^2 F_H(\rho)$ follows from theorem \ref{Thm:mono} in Sec.(\ref{Sec:Fisher-mon}) which implies Quantum Fisher Information cannot increase in the iid regime.

\end{document}